\documentclass{article}

\usepackage[left=3cm,right=3cm,top=2.5cm,bottom=2.5cm]{geometry}

\usepackage{amssymb,amsfonts,amsmath}
\usepackage{amsthm,xypic,enumerate}
\usepackage{mathtools, color}
\usepackage[sort&compress,comma,square,numbers]{natbib}
\usepackage{gensymb}
 
\usepackage[utf8]{inputenc}
\usepackage{multirow} 
\usepackage{hyperref}

\usepackage[version=3]{mhchem}

\usepackage{pgf,pgfarrows,pgfnodes}
\usepackage{tikz}
\usetikzlibrary{shapes}

\theoremstyle{plain}
\newtheorem{theorem}{Theorem}[section]
\newtheorem{proposition}[theorem]{Proposition}

\newtheorem{lemma}[theorem]{Lemma}

\theoremstyle{definition}

\newcommand{\R}{\mathbb{R}}

\DeclareMathOperator{\sign}{sign}

\begin{document}

\title{A proof of unlimited multistability for phosphorylation cycles  }
\author{Elisenda Feliu$^1$, Alan D. Rendall$^2$ and Carsten Wiuf$^3$}
\date{\today}

\footnotetext[1]{Department of Mathematical Sciences, University of Copenhagen, Universitetsparken 5, 2100 Copenhagen, Denmark. efeliu@math.ku.dk}
\footnotetext[2]{Institut für Mathematik, Johannes Gutenberg-Universität Mainz, Staudingerweg 9,
D-55099 Mainz, Germany. rendall@uni-mainz.de}
\footnotetext[3]{Department of Mathematical Sciences, University of Copenhagen, Universitetsparken 5, 2100 Copenhagen, Denmark. wiuf@math.ku.dk}

\maketitle

\begin{abstract}
The multiple futile cycle is a phosphorylation system in which a molecular substrate might be phosphorylated sequentially $n$ times by means of an enzymatic mechanism. The system has been studied mathematically using reaction network theory and ordinary differential equations. It is known that the system might have at least as many as $2\lfloor \tfrac{n}{2}\rfloor+1$ steady states (where $\lfloor x\rfloor$ is the integer part of $x$) for particular choices of  parameters. Furthermore, for the simple and dual futile cycles ($n=1,2$) the stability of the steady states has been determined in the sense that the only steady state of the simple futile cycle is globally stable, while there exist parameter values for which the dual futile cycle admits two asymptotically stable and one unstable steady state. For general $n$, evidence that the possible number of asymptotically stable steady states increases with $n$ has been given, which has led to the conjecture that parameter values can be chosen such that $\lfloor\tfrac{n}{2}\rfloor+1$ out of $2\lfloor\tfrac{n}{2}\rfloor+1$ steady states are asymptotically stable and the remaining steady states are unstable. 

We prove this conjecture here by first reducing the system to a smaller one, for which we find a choice of parameter values that give rise to a unique steady state with multiplicity $2\lfloor\tfrac{n}{2}\rfloor+1$. Using arguments from geometric singular perturbation theory, and a detailed analysis of the centre manifold of this steady state, we achieve the desired result.
\end{abstract}

\section{Introduction}

Post-translational modifications of proteins are ubiquitous in almost all molecular processes at the cellular level. Perhaps the most important post-translational modification process is that of phosphorylation, where no, one or several phosphate groups are attached to specific sites of a protein, thereby creating different protein phosphoforms.  It is estimated that more than one third  of all proteins in eukaryotes  are temporarily phosphorylated \cite{cohen}. Moreover, the number of phosphorylation sites varies greatly from protein to protein and, for example, exceeds $n=20$ in the case of the tumour suppressor protein p53 \cite{p53} and $n=80$ in the case of the tau protein, a microtubule stabiliser  \cite{tau}. However, the phosphorylation process appears in many different guises. One particularly widespread variant is that of distributive sequential phosphorylation where a protein is phosphorylated one site at a time and in  sequential order. Dephosphorylation is the reverse process and removes phosphate groups in the reverse order.

The standard model of distributive sequential phosphorylation consists of $6n$ molecular reactions
\begin{align*}
{\rm X}_{i-1} + {\rm E} &\ce{<=>} {\rm Y}_{1,i} \ce{->} {\rm X}_i + {\rm E}, & {\rm X}_i + {\rm F} &\ce{<=>} {\rm Y}_{2,i} \ce{->} {\rm X}_{i-1}+ {\rm F}, & i=1,\dots,n,
\end{align*}
where $n$ denotes the number of  phosphorylation sites, ${\rm X}_i$, $i=0,\ldots,n$, denotes the phosphoform with $i$ phosphate groups attached, ${\rm E}$ is an enzyme (a kinase) that catalyses the phosphorylation process, ${\rm F}$ is another enzyme (a phosphatase) that catalyses the dephosphorylation process, and ${\rm Y}_{1,i}$ and ${\rm Y}_{2,i}$, $i=1,\ldots,n$, are enzyme-substrate complexes. The model is also known as the multiple \textit{futile cycle} \cite{Wang:2008dc}.

Assuming mass-action kinetics \cite{erdi-toth}, the  reaction network  might be modelled by a system of  ordinary differential equations (ODEs) in the concentrations of the substances (the phosphoforms, enzyme-substrate complexes and free enzymes). It leads to a high-dimensional ODE system in $3n+3$ variables (concentrations) and $6n$ constants (one for each reaction). It has been argued, but not proved, that the number of possible \textit{stable} steady states of the system increases linearly with $n$, thereby allowing the cell a large amount of functional flexibility. The phenomenon has been named  \textit{unlimited multistability } \cite{TG-Nature,Wang:2008dc}. 

In the case $n=1$, the simple futile cycle, it is known that there is a unique  positive steady state for fixed total amounts of substrate (protein) and enzymes, and  that this solution is globally asymptotically stable \cite{angeli06}. In the  case $n=2$, the dual futile cycle, the conclusion was first reached with the help of  simulations that for certain choices of parameters there are three steady states, two of which are stable \cite{Markevich-mapk}. It was later proved rigorously   that  at most three steady states exist and that for certain parameter values there are indeed three \cite{Wang:2008dc}. In that paper nothing was proved about the stability of these  solutions. Later it was formally proven that for suitable choices of  parameters and total amounts of substrate and enzymes these three solutions are hyperbolic, two are asymptotically stable and the third is a saddle \cite{hell15a}. The  aim of this paper is to extend these results on stability to general values of  $n$.

For general $n$,  there can be more than two stable steady states \cite{G-PNAS}. Evidence that there can be as many as $2\lfloor\frac{n}{2}\rfloor+1$ steady states with $\lfloor\frac{n}{2}\rfloor+1$ of them being stable  was presented in 
\cite{TG-Nature}. Here $\lfloor x\rfloor$ denotes the floor function, the largest integer smaller than a real number $x$. In more detail, it was 
shown that steady states are in one to one correspondence with the intersections of two curves in the plane and these curves were plotted 
numerically. It was also indicated how in a certain formal limit these  intersections correspond to the roots of a polynomial in one variable. 
Starting from these considerations it has been  proved analytically  that there are parameters for which the number of steady 
states indicated in \cite{TG-Nature} exist but their stability was  not treated \cite{Wang:2008dc}. 

The number of steady states might be larger than $2\lfloor\frac{n}{2}\rfloor+1$ for certain choices of parameter values.  For $n=2$, the bound $2\lfloor\frac{n}{2}\rfloor+1$ is an upper bound, but for $n=3$, there can be as many as five steady states \cite{flockerzi14}.

We  prove that there exist parameter values for  which the ODE system has  $\lfloor\frac{n}{2}\rfloor+1$ stable and $\lfloor\frac{n}{2}\rfloor$ unstable
steady states (for fixed total amounts of substrate and enzymes).   The line of argument is the following. First we will  apply geometric singular perturbation theory (GSPT)  to reduce the ODE system in $3n+3$ variables to an ODE system in $n+1$ variables, the phosphoform concentrations. This system is known as the Michaelis-Menten limit (or system) of the original system \cite{hell15a}. For the Michaelis-Menten system we will show that there exists a choice of parameter values for which there is one asymptotically stable steady state with multiplicity $2\lfloor\frac{n}{2}\rfloor+1$.  
To achieve this, we apply a combination of linear algebra and dynamical systems theory. The main hurdle is to establish asymptotic stability of the steady state  as it is not hyperbolic but has a centre manifold of dimension one (for fixed total amount of substrate). The leading non-zero coefficient of the system in the direction of the centre manifold is negative but depends on $n$. Having established this, we again apply perturbation theory to conclude the existence of parameter values for which there are $2\lfloor\frac{n}{2}\rfloor+1$ steady states of which $\lfloor\frac{n}{2}\rfloor+1$ are asymptotically stable and 
$\lfloor\frac{n}{2}\rfloor$ are unstable. With this in place, we finally lift the steady states to the original system while preserving their stability properties.

\section{The basic equations}
We consider a protein ${\rm X}$ and let ${\rm X}_i$ be the protein with $i=0,\ldots,n$
phosphate groups attached to it. The kinase which catalyses the phosphorylation of ${\rm X}$ is 
denoted by ${\rm E}$ and the phosphatase which catalyses the dephosphorylation of ${\rm X}$ by ${\rm F}$.  The phosphorylation and dephosphorylation of the protein proceeds through the formation of enzyme-substrate complexes. For $i=1,\dots,n$, we let ${\rm Y}_{1,i}$ denote the enzyme-substrate complex formed by ${\rm X}_{i-1}$ and ${\rm E}$, and similarly by  ${\rm Y}_{2,i}$ the enzyme-substrate complex formed by ${\rm X}_i$ and ${\rm F}$. 
The reaction mechanism under consideration consists of  the reactions
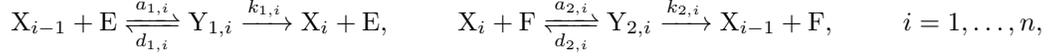
\begin{align*}
{\rm X}_{i-1} + {\rm E} &\ce{<=>[a_{1,i}][d_{1,i}]} {\rm Y}_{1,i} \ce{->[k_{1,i}]} {\rm X}_i + {\rm E}, & {\rm X}_i + {\rm F} &\ce{<=>[a_{2,i}][d_{2,i}]} {\rm Y}_{2,i} \ce{->[k_{2,i}]} {\rm X}_{i-1} + {\rm F}, & i=1,\dots,n,
\end{align*}
where the labels of the reactions indicate the reaction rate constants. These are positive real numbers.
Let  $x_i$ be the concentration of ${\rm X}_i$ for $i=0,\ldots,n$, $y_{1,i}$ and $y_{2,i}$ the 
concentrations of ${\rm Y}_{1,i}$ and ${\rm Y}_{2,i}$, respectively, 
for $i=1,\ldots,n$, and  let $x_{\rm E}$ and $x_{\rm F}$ be the concentrations of ${\rm E}$ and  ${\rm F}$, respectively. Under the assumption of mass-action kinetics, the evolution equations become 
\begin{align}
 \frac{dx_i}{dt}&=-a_{1,i+1}x_i x_{\rm E} -a_{2,i}x_ix_{\rm F}+d_{1,i+1}y_{1,i+1}+d_{2,i}y_{2,i}  +k_{1,i}y_{1,i}+k_{2,i+1}y_{2,i+1},& i&=0,\dots,n, \label{ev1}\\
 \frac{dy_{1,i}}{dt} & =a_{1,i}x_{i-1} x_{\rm E}-(d_{1,i}+k_{1,i})y_{1,i},& i&=1,\dots,n,\label{ev2}\\
 \frac{dy_{2,i}}{dt} & =a_{2,i}x_i x_{\rm F}-(d_{2,i}+k_{2,i})y_{2,i},& i&=1,\dots,n, \label{ev3}\\
 \frac{dx_{\rm E}}{dt} & =\sum_{i=1}^{n} -a_{1,i}x_{i-1} x_{\rm E} + (d_{1,i}+k_{1,i})y_{1,i}, \label{ev4}\\
 \frac{dx_{\rm F}}{dt} & =\sum_{i=1}^n -a_{2,i}x_i x_{\rm F} + (d_{2,i}+k_{2,i})y_{2,i},
\label{ev5}
\end{align}
where $a_{2,0}=d_{2,0}=k_{1,0}=0$, $a_{1,n+1}=d_{1,n+1}=k_{2,n+1}=0$.
The 
quantities
\begin{align}
E_{\rm tot} & =x_{\rm E}+\sum_{i=1}^{n} y_{1,i}\label{etotal},\\
F_{\rm tot} & =x_{\rm F}+\sum_{i=1}^n y_{2,i},\label{ftotal}\\
X_{\rm tot}& =\sum_{i=0}^n x_i+\sum_{i=1}^{n} (y_{1,i}+ y_{2,i}) \label{xtotal}
\end{align}
are conserved. Here $E_{\rm tot}$, $F_{\rm tot}$ and $X_{\rm tot}$ are known as the total amounts of kinase, phosphatase and substrate, respectively.  Note that using the equation for the conserved quantities $E_{\rm tot}$ and $F_{\rm tot}$, we might eliminate the variables $x_{\rm E}$ and  $x_{\rm F}$ from the right hand side of equations (\ref{ev1})-(\ref{ev5})  and thus the evolution equations for these two variables can be discarded. The relation (\ref{xtotal}) might be used to eliminate a further variable but this will not be done
here since it would destroy the symmetry of the system.

We next introduce rescaled variables that  generalise the rescaling used in the case of the dual futile cycle \cite{hell15a}. For a parameter 
$\epsilon>0$, let $x_{\rm E}=\epsilon\widetilde x_{\rm E}$, 
$x_{\rm F}=\epsilon\widetilde x_{\rm F}$, $y_{j,i}=\epsilon\widetilde y_{j,i}$ and 
$\tau=\epsilon t$. Denote the derivative with respect to $\tau$ by a prime 
and drop the tildes for convenience. This leads to the equations
\begin{align}
 \frac{dx_i}{d\tau}&=-a_{1,i+1}x_i x_{\rm E} -a_{2,i}x_ix_{\rm F}+d_{1,i+1}y_{1,i+1}+d_{2,i}y_{2,i}  +k_{1,i}y_{1,i}+k_{2,i+1}y_{2,i+1},\quad i=0,\dots,n, \label{ev21} \\
\epsilon \frac{dy_{1,i}}{d\tau} & =a_{1,i}x_{i-1} x_{\rm E}-(d_{1,i}+k_{1,i})y_{1,i},\qquad i=1,\dots,n,  \label{ev22}\\
\epsilon \frac{dy_{2,i}}{d\tau} & =a_{2,i}x_i x_{\rm F}-(d_{2,i}+k_{2,i})y_{2,i},\qquad i=1,\dots,n,  \label{ev23}\\
\epsilon \frac{dx_{\rm E}}{d\tau} & =\sum_{i=1}^{n} -a_{1,i}x_{i-1} x_{\rm E} + (d_{1,i}+k_{1,i})y_{1,i}, \label{ev24} \\
 \epsilon\frac{dx_{\rm F}}{d\tau} & =\sum_{i=1}^n -a_{2,i}x_i x_{\rm F} + (d_{2,i}+k_{2,i})y_{2,i}.  \label{ev25}
\end{align}
Correspondingly, we get rescaled conserved quantities $E_{\rm tot}$ and $F_{\rm tot}$ identical to \eqref{etotal} and \eqref{ftotal} (after cancelling $\epsilon$),  and the conserved quantity
\begin{equation*}\label{xtotalepsilon}
X_{\rm tot}=\sum_{i=0}^n x_i+\epsilon\sum_{i=1}^{n} (y_{1,i}+ y_{2,i}).
\end{equation*}

The expressions in the rescaled variables  remain regular as $\epsilon$ tends 
 to zero, in the sense that the functions on the right
hand sides of the equations extend in a $\mathcal{C}^\infty$-manner to $\epsilon=0$, 
and in the limit some of the evolution equations reduce to 
algebraic equations.  This is an example of the standard situation considered in GSPT. (For an introduction to GSPT see \cite{kuehn15}[Chapter 1].) In general we have a system of equations of the
form 
\begin{align*}
x'&=f(x,y,\epsilon), \\
\epsilon y'&=g(x,y,\epsilon). 
\end{align*}
Here the prime denotes the derivative with respect to $\tau$ and $\epsilon$ is 
a parameter. Introducing a new time variable $t=\tau/\epsilon$ and denoting the 
derivative with respect to $t$ by a dot, leads to the \textit{extended system}
\begin{align}
\dot x &=\epsilon f(x,y,\epsilon), \nonumber\\ 
\dot y&=g(x,y,\epsilon),\label{gsptextended2}\\
\dot \epsilon &=0. \nonumber
\end{align}
The variable $x$ is generally referred to as the \textit{slow} variable and $y$ as the \textit{fast} variable.

Assume that the equation $g(x,y,0)=0$ is equivalent to $y=h_0(x)$ for a 
continuously differentiable function $h_0$ so that the zero set of $g$ is a manifold. Then for $\epsilon=0$ the two equations for $x$ and $y$ in \eqref{gsptextended2} are equivalent to the system  $x'=f(x,h_0(x),0)$, which we refer to as the \emph{limiting system}. Of particular importance are the eigenvalues of the derivative of $g$ with  respect to $y$. In this context these  are known  as \emph{transverse eigenvalues}. 

The 
central conclusion is that when none of the transverse eigenvalues has zero 
real part, there exists an invariant manifold for the extended system called
the \emph{slow manifold} with the following properties:
\begin{itemize}
\item Its restriction to $\epsilon=0$ is the zero set of $g$. 
\item The restriction of the extended system to the invariant manifold is a system 
which, when written in terms of the time variable $\tau$, depends in a regular
manner on $\epsilon$ and agrees with the limiting system for $\epsilon=0$. 
\end{itemize}
The meaning of the word `regular' here is not only that the right hand sides of
the equations are functions of the parameter which for $\epsilon=0$ are as 
differentiable as the original system but also
that the equations can be solved for the time derivatives.

Consider now the evolution equations \eqref{ev21}-\eqref{ev25} for the multiple futile cycle with
the variables $x_{\rm E}$ and $x_{\rm F}$ eliminated using the rescaled  conservation equations, the analogues of \eqref{etotal} and \eqref{ftotal}. Then the 
variables $x_i$ might be taken as the slow variables $x$ in the general set-up and the variables $y_{j,i}$ as the fast variables $y$. 
Here, the equation $g(x,y,0)=0$ is equivalent to setting the right hand side of \eqref{ev2} and \eqref{ev3} to zero and eliminating  $x_{\rm E}$ and $x_{\rm F}$ using the conservation equations. 
By setting \eqref{ev2} and \eqref{ev3} to zero we obtain
\begin{equation*}\label{eq:y}
  y_{1,i} =  K_{1,i}^{-1} x_{i-1} x_{\rm E} ,\qquad y_{2,i} =K_{2,i}^{-1} x_i x_{\rm F}, 
  \end{equation*}
where 
\begin{equation}\label{eq:KJI}
K_{j,i}=\tfrac{d_{j,i}+k_{j,i}}{a_{j,i}},\quad\text{ for}\quad j=1,2, \quad i=1,\dots,n.
\end{equation}
 Using the conserved quantities $E_{\rm tot}$ and $F_{\rm tot}$, we further have by insertion
$$E_{\rm tot} = x_{\rm E} \left( 1 + \sum_{\ell=1}^n K_{1,\ell}^{-1} x_{\ell-1} \right), \qquad  F_{\rm tot} = x_{\rm F} \left( 1 + \sum_{\ell=1}^n K_{2,\ell}^{-1} x_{\ell} \right),  $$
which leads to
\begin{equation}\label{eq:y2}
  y_{1,i} = \frac{ K_{1,i}^{-1} E_{\rm tot} x_{i-1}}{1 + \sum_{\ell=1}^n K_{1,\ell}^{-1} x_{\ell-1}},\qquad y_{2,i} =\frac{K_{2,i}^{-1} F_{\rm tot} x_i}{1 + \sum_{\ell=1}^n K_{2,\ell}^{-1} x_{\ell}}
 \end{equation}
 for $i=1,\ldots,n$.
These expressions define the function $h_0$ above. 

The transverse  eigenvalues are by definition the eigenvalues of the 
linearisation $L$ of the right hand side of the evolution equations for the 
concentrations of the substrate-enzyme complexes with respect to those 
complexes. That is, the linearisation of the equations for $i=1,\dots,n$,
\begin{align*}
& a_{1,i}x_{i-1} (E_{\rm tot} -\sum_{\ell=1}^{n} y_{1,\ell}) -(d_{1,i}+k_{1,i})y_{1,i}  \\
& a_{2,i}x_i (F_{\rm tot} -\sum_{\ell=1}^{n} y_{2,\ell})-(d_{2,i}+k_{2,i})y_{2,i}.
\end{align*}
The matrix $L$ can be seen as the sum of a diagonal matrix $L_1$ and another matrix 
$L_2$. The diagonal elements of $L_1$ are $-(d_{1,i}+k_{1,i})$ for the rows corresponding to $y_{1,i}$, $i=1,\ldots, n$, and $-(d_{2,i}+k_{2,i})$ for the rows corresponding to $y_{2,i}$, $i=1,\ldots, n$. The matrix $L_2$ is 
block diagonal with one block corresponding to each of the two enzymes. We 
only need to consider one of these blocks as the other one can be treated in a strictly analogous manner. Consider the block corresponding to the enzyme ${\rm E}$.  All entries of the $i$-th row of the block  are equal to $-a_{1,i}x_{i-1}$. The transverse eigenvalues have negative real
part because all eigenvalues of $-L$ have positive real part, as shown in the lemma below.

\begin{lemma}\label{lemma:Pmatrix}  
If $M=(m_{ij})$ is a matrix with elements of the form $m_{ij}=a_{ij}+b_i$
where $a_{ij}=0$ for all $i\ne j$ and all $a_{ii}$ and $b_i$ are positive,
then all eigenvalues of $M$ have positive real parts.
\end{lemma}

The proof  is given in Subsection~\ref{sec:Pmatrix}.

\medskip
We conclude that GSPT applies to this situation and that there exists a slow manifold. The restriction of the extended system to the slow manifold gives a  family of dynamical systems depending regularly on the parameter $\epsilon$. The level sets of the conserved quantity defined by $X_{\rm tot}$ are transverse to the slow manifold and so restricting to a constant value of  this conserved quantity also results in a regular parameter-dependent  dynamical system, which we call the \emph{completely reduced system}. Note that  the dimension of the restriction of the extended system to the slow manifold is  greater by one than that of the completely reduced system.

At this point we need some material from the theory of dynamical systems.
Background on this can be found in \cite{perko}. Let $\dot x=f(x)$ be a system
of ODEs on $\R^n$, where $f$ is $C^1$. Let $x_0$ be a steady state of this 
system, that is, a point where $f(x_0)=0$. Let $A=Df(x_0)$ denote the 
derivative (Jacobian) of the right hand side of the equations at $x_0$. Then $\R^n$ can be written as the direct sum of three linear
subspaces $E_-$, $E_c$ and $E_+$, which are spanned by the real and imaginary 
parts of the generalised eigenvectors of $A$ corresponding to the eigenvalues whose
real parts are negative, zero and positive respectively. These subspaces are 
called the stable, centre and unstable subspaces at $x_0$. If there are no 
eigenvalues with zero real part, so that $E_c$ is trivial, then the point $x_0$
is said to be hyperbolic. 
Suppose now that we have a parameter-dependent 
system $\dot x=f(x,\alpha)$ where $f(x_0,0)=0$ and $f$ is $\mathcal{C}^1$ in its 
dependence on both $x$ and $\alpha$. The point $x_0$ is a steady 
state of the system for $\alpha=0$. If this steady state is hyperbolic, then it
follows that for $\alpha$ close to zero there exists a unique steady state
close to $x_0$ and that it is hyperbolic. The dimensions of the stable and 
unstable subspaces of this steady state are independent of $\alpha$ for 
$\alpha$ small.
 
Returning to our concrete example, if the limiting system has $k$ hyperbolic 
steady states, then so does the system defined by applying all three conservation laws to the original system for $\epsilon$ small (playing the role of $\alpha$ here). The stability type of these steady states is preserved, in the sense that the dimensions of the stable and 
unstable manifolds are independent of $\epsilon$. In particular, if for one of the steady states of the limiting system, the stable subspace is the whole of $\R^n$, then this is
also true for the corresponding steady state with $\epsilon>0$. Moreover, this is known to imply that the steady state is asymptotically stable. In general a steady state of the restriction of the system to the slow manifold is a steady state of the extended system and the dimension of its stable manifold in the extended system is the sum of the dimension of its stable manifold in the extended system  restricted to the slow manifold and the number of transverse eigenvalues with  negative real part.
 
Using the fact that the right hand sides of \eqref{ev22} and \eqref{ev23} are zero for $\epsilon=0$ and the expressions in \eqref{eq:y2}, it follows that  
the limiting system, which in this case will also be referred to as the Michaelis-Menten  (or MM) system, consists of the equations 
\begin{align*}
\frac{dx_i}{d\tau} & =-\frac{k_{1,i+1}K_{1,i+1}^{-1}E_{\rm tot}x_i }{1 + \sum_{\ell=1}^n K_{1,\ell}^{-1} x_{\ell-1}}
+
\frac{k_{1,i}K_{1,i}^{-1}E_{\rm tot}x_{i-1}}
{1 + \sum_{\ell=1}^n K_{1,\ell}^{-1} x_{\ell-1}} \\ & \qquad \qquad
-\frac{k_{2,i}K_{2,i}^{-1}F_{\rm tot}x_i}{1+\sum_{\ell=1}^{n} K_{2,\ell}^{-1}x_\ell} 
+\frac{k_{2,i+1}K_{2,i+1}^{-1}F_{\rm tot}x_{i+1}}
{1+\sum_{\ell=1}^{n} K_{2,\ell}^{-1}x_\ell}
\end{align*}
for $i=0,\ldots, n$, where we adopt  the convention that the symbols $K_{2,0}^{-1}$, $K_{1,0}^{-1}$, $K_{1,n+1}^{-1}$ and $K_{2,n+1}^{-1}$ are defined to be zero. 

The rest of the paper is devoted to proving the following theorems.

\begin{theorem}\label{thm:MM}
There exists a choice of positive parameters (reaction rate constants and $X_{\rm tot}$) such that the Michaelis-Menten system admits  $2\lfloor\frac{n}{2}\rfloor+1$ steady states in the linear invariant subspace defined by the total amount $X_{\rm tot}$. Furthermore, relatively to this invariant subspace, $\lfloor\frac{n}{2}\rfloor+1$ of these steady states are asymptotically stable and hyperbolic, and $\lfloor\frac{n}{2}\rfloor$ are unstable and hyperbolic with a one-dimensional  unstable manifold.
\end{theorem}

As a consequence, GSPT allows us to conclude that the original system with evolution equations  \eqref{ev1}-\eqref{ev5} also admits a choice of positive parameters such that  there are $\lfloor\frac{n}{2}\rfloor+1$ asymptotically stable steady states and $\lfloor\frac{n}{2}\rfloor$ unstable steady states. Indeed, given that the transverse  eigenvalues are all negative this means that for $\epsilon$ small there exist steady states of the original system corresponding to the steady states of the MM system. They are hyperbolic and have corresponding stability properties. The
sinks remain sinks and the dimension of the unstable manifolds of the other  steady states remains one. This completes the proof of the theorem below.

\begin{theorem}\label{thm:full}
There exists a choice of positive parameters (reaction rate constants and total amounts) such that the  $n$-site phosphorylation system with evolution equations \eqref{ev1}-\eqref{ev5}  admits  $2\lfloor\frac{n}{2}\rfloor+1$ steady states in the linear invariant subspace defined by the total amounts $E_{\rm tot}$, $F_{\rm tot}$ and $X_{\rm tot}$. Furthermore, relatively to this invariant subspace, $\lfloor\frac{n}{2}\rfloor+1$ of these steady states are asymptotically stable and hyperbolic, and $\lfloor\frac{n}{2}\rfloor$ are unstable and hyperbolic with a one-dimensional  unstable manifold.
\end{theorem}

\section{Proof of Theorem~\ref{thm:MM}}

This section is devoted to selecting `nice' parameters such that the MM system has  $2\lfloor\tfrac{n}{2}\rfloor+1$ positive steady states with the same $X_{\rm tot}$, of which  $\lfloor\frac{n}{2}\rfloor+1$ are asymptotically stable. To achieve this, we first select parameter values such that $(1,\dots,1)\in \R^{n+1}_{>0}$ is the only positive steady state of the MM system and has multiplicity $2\lfloor\tfrac{n}{2}\rfloor+1$. We then study the stability of this steady state.  To this end, we need some centre manifold theory and this will now be reviewed. Background on this can be found in \cite{carr81} and \cite{kuznetsov95}. Consider once again the system $\dot x=f(x)$ and a steady state $x_0$. Suppose that $f$ is of class $\mathcal{C}^k$ for some finite $k\ge 1$. We restrict consideration to a small neighbourhood of $x_0$. There exist manifolds $V_-$, $V_c$ and $V_+$ of class $\mathcal{C}^k$ passing through $x_0$ that are invariant under the flow  generated by the differential equations and are tangent to $E_-$, $E_c$ and  $E_+$, respectively, at $x_0$. They are referred to as stable, centre and  unstable manifolds of the system at $x_0$. $V_-$ and $V_+$ are unique but  $V_c$ is in general not unique. Fortunately, this lack of uniqueness is  usually not a problem in applications as  will also be illustrated in our case (to be argued later). It may be noted in passing that the analogues of the statements just made are in general not true if $\mathcal{C}^k$ with $k$ finite is replaced everywhere by $\mathcal{C}^\infty$ in the sense that there might not exist a $\mathcal{C}^\infty$ centre manifold  \cite[Section 5.1]{kuznetsov95}.

We show that the Jacobian of the MM system evaluated at the steady state has $n-1$ eigenvalues with negative real part and two zero eigenvalues. Due to the conservation equation for $X_{\rm tot}$, this implies that in the system defined by restriction to a fixed value of $X_{\rm tot}$, this point has a one-dimensional centre manifold. We proceed to study this centre manifold and conclude that the steady state is asymptotically stable. Finally, we use a perturbation argument to ensure that modified parameters can be found for which there are $2\lfloor\frac{n}{2}\rfloor+1$ positive steady states, of which  $\lfloor\frac{n}{2}\rfloor+1$ are asymptotically stable and the remaining ones unstable.

In the following we assume $n\ge 2$, and note that the case $n=1$ has already been solved \cite{angeli06}. The outline of this section is as follows:
\begin{itemize}
\item[\S \ref{sec:ss-multi}] We find parameter values such that the limiting system has one steady state. Furthermore, we give the eigenvalue structure of the Jacobian evaluated at this steady state.
\item[\S \ref{sec:centre}] The centre manifold is studied. We show that the ODE system of the limiting system is attractive towards the steady state along the direction of the centre manifold. Combined with the results of Subsection \ref{sec:ss-multi}, this establishes the asymptotic stability of the steady state.
\item[\S \ref{sec:conclude}] The parameter values are perturbed to establish the existence of $\lfloor\frac{n}{2}\rfloor+1$  asymptotically stable steady states of the MM system and these are then lifted to the original system. 
\end{itemize}

\subsection{A steady state with multiplicity $2\lfloor\tfrac{n}{2}\rfloor+1$}
\label{sec:ss-multi}

We let
$$X_{\rm tot}=n+1,\quad  F_{\rm tot}=E_{\rm tot}=1,\qquad k_{1,i}=K_{1,i},\quad k_{2,i}=K_{2,i},\ i=1,\dots,n.$$
 To simplify the notation, we define
$$ \alpha_i = K_{1,i}^{-1},  \qquad \beta_i = K_{2,i}^{-1},\quad i=1,\dots,n.$$
Given arbitrary positive values of $K_{j,i}$ one can  always find positive values of $d_{j,i}, a_{j,i}, k_{j,i}$ corresponding to $K_{j,i}$, see \eqref{eq:KJI}. We will therefore be concerned with choosing $K_{j,i}$. 

With these definitions the MM system becomes:
\begin{align}
\frac{dx_0}{d\tau} & =-\frac{x_0}{1+\sum_{\ell=1}^{n}  \alpha_\ell x_{\ell-1}}
+\frac{ x_{1}}{1+\sum_{\ell=1}^{n}  \beta_\ell x_\ell},  \nonumber \\
\frac{dx_i}{d\tau} & =\frac{x_{i-1} -x_i}{1+\sum_{\ell=1}^{n}  \alpha_\ell x_{\ell-1}}+ 
\frac{x_{i+1}-x_i}{1+\sum_{\ell=1}^{n}  \beta_\ell x_\ell}, \qquad i=1,\dots,n-1, \label{eq:MMgood}\\
\frac{dx_n}{d\tau} & =\frac{ x_{n-1}}
{1+\sum_{\ell=1}^{n}  \alpha_\ell x_{\ell-1}}-\frac{x_n}{1+\sum_{\ell=1}^{n}  \beta_\ell x_\ell} .\nonumber
\end{align}
Recall that the sum of these equations is zero. 

We start by reducing the steady state system to a polynomial equation in one variable. We equate  the right-hand side of \eqref{eq:MMgood} to zero. 
Let 
\begin{equation}\label{eq:u}
 u = \frac{1+\sum_{\ell=1}^{n}  \beta_\ell x_\ell}{1+\sum_{\ell=1}^{n}  \alpha_\ell x_{\ell-1}}.
 \end{equation}
By setting the first equation in \eqref{eq:MMgood} to zero, we obtain  $x_1 = u x_0$. Using this relation and $\frac{dx_1}{d\tau}=0$ we obtain $x_2= u\, x_1 = u^2 x_0$. By iteration we obtain
\begin{equation}\label{eq:xiu}
x_i = u^i x_0,\qquad i=0,\dots,n.
\end{equation}
The assumption on the conserved quantity, $X_{\rm tot}=n+1$, implies
\begin{equation}\label{eq:x0}
 n+1 = x_0( 1+u+\dots + u^n),\qquad \textrm{hence} \qquad x_0 = \frac{n+1}{1+u+\dots + u^n}.
\end{equation}
Using \eqref{eq:u} and \eqref{eq:xiu}  we find 
$$1+\sum_{\ell=1}^{n}  \beta_\ell u^\ell x_0 = u+\sum_{\ell=1}^{n} \alpha_\ell u^{\ell} x_0, $$
which after substitution of the value for $x_0$ in \eqref{eq:x0} and multiplication by $1+u+\dots + u^n$ gives
$$1+u+\dots + u^n+\sum_{\ell=1}^{n}  \beta_\ell u^\ell (n+1)= u\, (1+u+\dots + u^n)+\sum_{\ell=1}^{n} \alpha_\ell u^{\ell}(n+1).  $$
Finally, rearranging the terms results in the following  polynomial equation in one variable:
 \begin{equation}\label{eq:thepoly} u^{n+1}  + (n+1)(\alpha_{n} - \beta_n)  u^n+\ldots + (n+1)(\alpha_{i} - \beta_i)  u^i+ \ldots
+(n+1)(\alpha_{1} - \beta_1) u - 1 =0.
\end{equation}
We conclude that the positive steady states of the MM system are in one-to-one correspondence with the positive solutions of the univariate polynomial equation in \eqref{eq:thepoly}. Given a solution $u$ of this equation, then $x_0,\dots,x_n$ are found from \eqref{eq:x0} and \eqref{eq:xiu}.

\medskip
For any choice of values of $\alpha_i,\beta_i$ such that 
\begin{equation}\label{eq:nicealpha}
 (n+1)(\alpha_{i} - \beta_i)  =\left\{\begin{array}{cl} (-1)^{i+1} \binom{n+1}{i}  & \quad\text{for}\quad n\quad\text{even}\\
(-1)^{i+1}\left(\binom{n}{i}-\binom{n}{i-1} \right) & \quad\text{for}\quad n\quad\text{odd} \end{array} \right. \qquad i=1,\dots,n,
\end{equation}
the polynomial on the left-hand side of \eqref{eq:thepoly} is $(u-1)^{n+1}$ for $n$ even and $(u-1)^n (u+1)$ for $n$ odd.  
So the only positive root is $u=1$, which in turn implies that $x_0=1$, and  the steady state is $p=(1,\dots,1)$. 
This solution has multiplicity $2\lfloor\tfrac{n}{2}\rfloor+1$. Note that positive $\alpha_i,\beta_i$ can always be found such that \eqref{eq:nicealpha} is satisfied.

We move on to show that $\alpha_i,\beta_i$  additionally can be chosen  such that the Jacobian of the MM system evaluated at $p$ has rank $n-1$ (hence the zero eigenvalue has multiplicity two)  and the non-zero eigenvalues have negative real part. One zero eigenvalue corresponds to the conserved quantity $X_{\rm tot}$. If the non-zero eigenvalues have negative real part, then the stability of the steady state might be understood from studying the behaviour of the system along the direction corresponding to the second zero eigenvalue.  This is done in Section \ref{sec:centre}. 
 
For this purpose we  introduce a combinatorial quantity. Define
\begin{equation}\label{eq:gamman}
 \gamma_{n}(i):=\left\{\begin{array}{cl} \frac{(-1)^{i+1}}{n+1} \binom{n}{i} & \quad\text{for}\quad n\quad\text{even,}\\[5pt]
\frac{(-1)^{i+1}}{n+1} \left(    \binom{n-1}{i} - \binom{n-1}{i-1}  \right)  =
\frac{(-1)^{i+1}(n-2i)}{n(n+1)} \binom{n}{i}  & \quad\text{for}\quad n\quad\text{odd,} \end{array} \right. 
\end{equation}
and $i=0,\ldots,n$. Note that  $\gamma_n(0)=\gamma_n(n)=-\tfrac{1}{n+1}$ for all $n$. 
Let  $\gamma_n=(\gamma_n(0),\ldots,\gamma_n(n))$. 

Furthermore, we recall a standard binomial identity \cite{ruiz}:
\begin{equation}\label{eq:binom0}
0 =\sum_{j=0}^n (-1)^j P(j) \begin{pmatrix} n \\ j \end{pmatrix}\qquad \text{for any polynomial }P \textrm{ of degree smaller than }n. 
\end{equation}
In particular, we have 
\begin{align*}
0& =\sum_{j=0}^n(-1)^j \binom{n}{j} \qquad  \text{for}\quad  n>0,
\\
0 & =\sum_{j=0}^n(-1)^j (j+i) \binom{n}{j} =\sum_{j=1}^n(-1)^j j\binom{n}{j} \qquad  \text{for any }i\in \mathbb{Z}\quad \text{ and}\quad  n>1. 
\end{align*}
Additionally, we will apply Pascal's recurrence:
\begin{equation}\label{eq:binom4}
\binom{n}{i-1}   + \binom{n}{i} =\binom{n+1}{i},
\end{equation}
which is valid for non-negative integers, $i,n\ge 0$.

\begin{proposition}\label{prop:good}
For any real $\alpha_1$, define
$$ \beta_i := \alpha_1 - \gamma_n(i) -\tfrac{1}{n+1}, \quad i=1,\dots,n, \qquad \textrm{and}\qquad \alpha_i:=\beta_{i-1},\qquad  i=2,\dots,n.$$
Then \eqref{eq:nicealpha} is satisfied and further, 
$$ \beta_n  =\alpha_1, \quad \textrm{and}\quad 
1+\sum_{i=1}^{n} \alpha_i = 1+\sum_{i=1}^{n} \beta_i= \alpha_1 n.
$$
\end{proposition}
\begin{proof}
By definition
$$ \beta_n = \alpha_1  - (-\tfrac{1}{n+1}) - \tfrac{1}{n+1} = \alpha_1.$$
For $i=1,\ldots,n$ the relation 
$\alpha_i=\alpha_1 - \gamma_n(i-1) -\tfrac{1}{n+1}$  holds since $\alpha_i=\beta_{i-1}$ for $i\geq 2$ and  $\gamma_n(0)=\frac{-1}{n+1}$. This gives
$$
\alpha_{i} - \beta_i =   - \gamma_n(i-1) +\gamma_n(i).
$$
For $n$ even we have
\begin{align*}
\alpha_{i} - \beta_i =  \frac{(-1)^{i+1}}{n+1} \begin{pmatrix} n \\  i-1 \end{pmatrix}+   \frac{(-1)^{i+1}}{n+1} \begin{pmatrix} n \\  i \end{pmatrix} =   \frac{(-1)^{i+1}}{n+1}\begin{pmatrix} n+1 \\  i \end{pmatrix},
\end{align*}
where  \eqref{eq:binom4} is applied. If $n$ is odd, then using the expression of $\gamma_n(i)$ \eqref{eq:gamman}, we obtain: 
\begin{align*}
\alpha_{i} - \beta_i =  \frac{(-1)^{i+1}}{n+1}   \left(\begin{pmatrix} n-1 \\  i-1 \end{pmatrix}-\begin{pmatrix} n-1 \\ i-2\end{pmatrix} + \begin{pmatrix} n-1 \\  i \end{pmatrix}-\begin{pmatrix} n-1 \\ i-1\end{pmatrix}\right)    =  \frac{(-1)^{i+1}}{n+1}\left(\begin{pmatrix} n \\ i \end{pmatrix}-\begin{pmatrix} n \\ i-1\end{pmatrix}\right)
\end{align*}
where  \eqref{eq:binom4} is used twice.
This concludes the proof of \eqref{eq:nicealpha}.

It remains to show that $1+\sum_{i=1}^{n} \alpha_i = 1+\sum_{i=1}^{n} \beta_i= \alpha_1 n$. The first equality is a consequence of definition of $\alpha_i$ and $\alpha_1=\beta_n$.  For the second equality, we have
\begin{align*}
 1+\sum_{i=1}^{n} \beta_i &= 1+\alpha_1 +\sum_{i=1}^{n-1}  \left(\alpha_1 - \gamma_n(i) -\frac{1}{n+1}\right)   = 1+ n\, \alpha_1 -  \frac{n-1}{n+1}  - \sum_{i=1}^{n-1}   \gamma_n(i). 
\end{align*}
If  $ (n+1)\sum_{i=1}^{n-1}   \gamma_n(i)=2$ then we are done. To show this, apply \eqref{eq:binom0}  with $P(j)= n-2j$ for $n$ odd and $P(j)=n$ for $n$ even (recalling that $n>1$):
\begin{align*} 
(n+1)\sum_{i=1}^{n-1}   \gamma_n(i) &=\tfrac{-1}{n(n+1)}
\sum_{i=1}^{n-1}  (-1)^i P(i) \binom{n}{i}= \tfrac{1}{n(n+1)}(  P(0) + (-1)^{n} P(n)) \\
&=\begin{cases}
\tfrac{2n}{n(n+1)} = \tfrac{2}{n+1} & \textrm{for $n$ even}, \\
\tfrac{ n - (-n)}{n(n+1)}  = \tfrac{2}{n+1} & \textrm{ for $n$ odd}.
\end{cases}
\end{align*}
This concludes the proof.
\end{proof}

\medskip
According to  Proposition~\ref{prop:good} and the discussion after \eqref{eq:nicealpha}, by choosing $\alpha_1>0$ large enough such that $\beta_i> 0$ for all $i=1,\dots,n$, then 
 $p=(1,\dots,1)$ is a steady state of the MM system. In the remaining part of the text, we consider $\alpha_i,\beta_i$ chosen such that this is the case.
 We proceed to study the Jacobian matrix  $J\in \R^{(n+1)\times (n+1)}$ of the MM system \eqref{eq:MMgood} evaluated at $p$, with this choice of parameters. 
 
We obtain the following partial derivatives of the MM system \eqref{eq:MMgood}, evaluated at $p$:
\begin{align*}
\frac{\partial }{\partial x_j}\left( \frac{x_i}{1+\sum_{\ell=1}^{n}  \alpha_\ell x_{\ell-1}}\right)\Big|_{x=p} & =
 \begin{cases}    \frac{1}{\alpha_1 n} -   \frac{\alpha_{j+1} }{(\alpha_1 n)^2}  &\text{for}\quad i=j, \\ 
 \frac{- \alpha_{j+1} }{(\alpha_1 n)^2} & \text{for}\quad i\neq j,
\end{cases}
\\
\frac{\partial }{\partial x_j}\left( \frac{x_i}{1+\sum_{\ell=1}^{n}  \beta_\ell x_{\ell}}\right)\Big|_{x=p} & =
 \begin{cases}    \frac{1}{\alpha_1 n} -   \frac{\beta_{j} }{(\alpha_1 n)^2}  &\text{for}\quad i=j, \\ 
 \frac{- \beta_{j} }{(\alpha_1 n)^2} & \text{for}\quad i\neq j,
\end{cases}
\end{align*}
where $\alpha_{n+1}=\beta_0=0$. 
The matrix $J$ might be reformulated in terms of two other matrices:
{\small \begin{align*}
\alpha_1 n J &=J_1+J_2 \\
&= \begin{pmatrix}   -1 & \phantom{-}1 & \phantom{-}0 & \ldots & \phantom{-}0 \\ \phantom{-}1 & -2 & \phantom{-}1 & \ldots & \phantom{-}0 \\ \vdots & \ddots & & \ddots & \vdots \\ \phantom{-}0 & \ldots & \phantom{-}1 & -2 & \phantom{-}1 \\ \phantom{-}0  & \ldots &  \phantom{-}0  & \phantom{-}1 & -1 \end{pmatrix}  + \frac{1}{\alpha_1 n } 
\begin{pmatrix}  
\alpha_1 &  \alpha_2-\beta_1 &  \alpha_3-\beta_2 & \dots &  \alpha_{n}-\beta_{n-1} &  - \beta_n \\
0 & 0 & 0 & \dots & 0 & 0 \\
\vdots & \vdots & \vdots & \vdots & \vdots & \vdots \\
0 & 0 & 0 & \dots & 0 & 0 \\
-\alpha_1 &  -(\alpha_2-\beta_1) &  -(\alpha_3-\beta_2) & \dots &  -(\alpha_{n}-\beta_{n-1}) &   \beta_n
\end{pmatrix}.
\end{align*}}
The rows $2,\dots,n$ of $\alpha_1 n\, J_2$ are zero because  the $j$-th entry is $-\alpha_{j} +  \alpha_{j}+\beta_{j-1}- \beta_{j-1}=0$.
Using Proposition~\ref{prop:good},   the matrix $\alpha_1 n J$ becomes  the following symmetric matrix:
\begin{equation}\label{eq:goodJ}
\widetilde{J}:=\alpha_1 n J = \begin{pmatrix}   -1 +\frac{1}{n} &  1 & 0 & \ldots & 0  & - \frac{1}{n}  \\ 
1 & -2 & 1 & \ldots & 0  & 0 \\ 0 &   1 & -2 & 1 & \ldots   & 0 \\    \vdots & \vdots & \ddots & \ddots & \ddots & \vdots  \\ 0 & 0 & \ldots & 1 & -2 & 1 \\  
-\frac{1}{n}  & 0 & \ldots &  0  & 1 & -1 + \frac{1}{n}\end{pmatrix}\in \R^{(n+1)\times (n+1)}.
\end{equation}
The  matrix $\widetilde J$ is independent of the choice of $\alpha_1$.  The rows $2,\dots,n$ of this matrix are clearly linearly independent.

\begin{proposition}\label{eq:eigenvalues}
Consider the matrix $\widetilde{J}$ in \eqref{eq:goodJ}.
\begin{itemize}
\item  $\widetilde{J}$ has rank $n-1$. An orthogonal basis of the kernel of $\widetilde J$ is formed by the vector $(1,\dots,1)$ and the vector $v=(v_0,\dots,v_n)\in \R^{n+1}$ with
$$v_i = n - 2i,\qquad i=0,\dots,n.$$
\item $\widetilde{J}$ has $n-1$ real negative eigenvalues, counted with multiplicity, and two zero eigenvalues.
\end{itemize}
\end{proposition}
\begin{proof} 
It is straightforward to see that  $(1,\dots,1)$ belongs to the kernel of $\widetilde{J}$, because the column sums are zero.
This vector is linearly independent of $v$. 
To show that $v$ also is in the kernel of $\widetilde{J}$ we note that the scalar product of $v$ with the rows $i=2,\dots,n$ is
$$ v_{i-2} -2v_{i-1} + v_{i} = n - 2(i-2)  - 2(n -2i+2) + n - 2i= 0.$$
For rows with index $1$ and $n+1$, we have
\begin{align*}
\left(-1+\frac{1}{n}\right) v_0 + v_1 - \frac{1}{n}v_n = -n+1 + n-2 -1(-1)& =0, \\
-\frac{1}{n}v_0 +v_{n-1}+\left(-1+\frac{1}{n}\right)v_n = -1 + (-n+2) +(n-1) & =0.
\end{align*}
Hence $v$ is  in the kernel of $\widetilde{J}$. Since the rank of $\widetilde{J}$ is at least $n-1$, $(1,\dots,1)$ and $v$ form a basis of the kernel.
A simple computation shows that the sum of the entries of $v$ is zero, that is, $(1,\ldots,1)\cdot v=0$:
$$ \sum_{i=0}^n ( n - 2i) = n(n+1) - 2  \sum_{i=0}^n  i =n(n+1) - 2   \tfrac{n(n+1)}{2} =0. $$
Hence the basis is orthogonal.

To study the eigenvalues of $\widetilde{J}$ we proceed as follows. Since $\widetilde{J}$ is real and symmetric, all eigenvalues are real. 
We already know that zero is an eigenvalue with multiplicity two. By Descartes' rule of signs, if all coefficients of the characteristic polynomial of $\widetilde{J}$ are non-negative, then the number of positive roots is zero. Hence the remaining $n-1$ roots of the characteristic polynomial (counted with multiplicity) must be negative.  

To show that  all coefficients of the characteristic polynomial of $\widetilde{J}$ are non-negative, we show that for $i=1,\dots,n-1$, all  non-zero principal minors of size $i$ of $\widetilde{J}$ have sign $(-1)^i$.
 Clearly, this holds for the principal minors of size $1$, since the diagonal entries of $\widetilde{J}$ are all negative (recall $n>1$).

Consider the determinant of a matrix of size $i$ of the form
\begin{equation}\label{eq:Ai}
 A_i:= \begin{pmatrix}
-a_1 & 1 & 0 & \dots & 0 \\
1 & -a_2 & 1 & \dots & 0 \\ 
\vdots & \ddots & \ddots & \ddots & \vdots \\
0 & \dots & 1 & -a_{i-1} & 1 \\
0 &   \dots & 0  & 1 & -a_i
\end{pmatrix}
\end{equation}
for the following cases: 
\begin{itemize}
\item[(1)] $a_j=2$ for all $j=1,\ldots,i$.
\item[(2)] $a_1=1-\tfrac{1}{n}$ and $a_j=2$ for $j>1$.
\item[(3)] $a_i=1-\tfrac{1}{n}$ and $a_j=2$ for $j<i$.
\end{itemize}
For the three cases, define respective vectors:  (1) $x=(1,\dots,1)$, (2) $x=(n,n-1,\dots,n-i+1)$, and (3) $x=(n-i+1,\dots,n-1,n)$. Then the components of  $-A_i x$ are non-negative. Using \cite[Theorem 5.4]{czechMath} on the matrix $-A_i$  and subsequently using \cite[Theorem 5.1 2\degree]{czechMath}, we conclude that the principal minors of $-A_i$ are non-negative, hence in particular  $\det(A_i)$ is either zero or has sign $(-1)^i$.

For a subset $H$ of $\{1,\dots,n+1\}$ of size $n+1-j$, $j=2,\ldots,n-1$, we consider  the submatrix $\widetilde{J}_H$ of size $j$ of $\widetilde{J}$  obtained by removing the columns and rows  with indices in $H$. If $H$ contains $1$ or $n+1$, then the submatrix 
$\widetilde{J}_H$ is a block matrix with  blocks of the form \eqref{eq:Ai}. 
Hence the  determinant of  $\widetilde{J}_H$ is the product of the determinants of the blocks, and it is either zero or has sign $(-1)^j$.

If $H$ contains neither   $1$ nor $n+1$, then $\widetilde{J}_H$ has the form
$$\widetilde{J}_H=\left(  \begin{array}{c|cccc|c}   -1 +\frac{1}{n} &  z_1 & 0 & \ldots & 0  & - \frac{1}{n}  \\ \hline
z_1 &  &  &  &   & 0 \\   
0 & & B &&& \vdots  \\  \vdots &  &  &  &  & 0 \\ 0 &  &  &  &  & z_2  \\  \hline
-\frac{1}{n}  & 0 & \ldots &  0  & z_2 & -1 + \frac{1}{n}\end{array}\right),   $$
where $z_1=1$ if $2\not\in H$  and $z_1=0$ otherwise, $z_2=1$ if $n\not\in H$ and $z_2=0$ otherwise, and $B$ is a block matrix of size $j-2$ with blocks of the form \eqref{eq:Ai}  with diagonal entries equal to $-2$.

If $z_1=z_2=0$ then $\det(\widetilde{J}_H)$ equals $\det(B)$ times the determinant of 
$$\begin{pmatrix} -1 +\frac{1}{n} & - \frac{1}{n}  \\ - \frac{1}{n}  & -1 +\frac{1}{n} \end{pmatrix}.$$
 Since the later is positive, the determinant has the desired sign. 

If $z_1=1$ and $z_2=0$ (and analogously for $z_1=0, z_2=1$), then we expand the determinant along the last column and obtain
$$
\det(\widetilde{J}_H)= (-1 + \tfrac{1}{n}) \det \left(  \begin{array}{c|cccc}   -1 +\frac{1}{n} &  1 & 0 & \ldots & 0    \\ \hline
1 &  &  &  &    \\   0 &  &B  &  &    \\  \vdots &  &  &  &   \\ 0 &  &  &  &    \end{array}\right) 
+(-1)^{j} \tfrac{1}{n} \det \left(  \begin{array}{c|cccc}  
1 &  &  &  &    \\   0 &  & B &  &    \\   \vdots &  &  &  &     \\ 0 &  &  &  &     \\  \hline
-\frac{1}{n}  & 0 & \ldots &  0  & 0 \end{array}\right).
$$
Note in this case we necessarily have $j\geq 3$. 
Let $B_1$ denote the matrix obtained by removing the first column and the first row of $B$. Then $B_1$ is a block matrix of size $j-3$ with blocks of the form \eqref{eq:Ai} and hence its determinant, if non-zero, has sign $(-1)^{j-3}$.
We then have
 $$\det(\widetilde{J}_H) =  (-1 + \tfrac{1}{n})^2  \det(B) - \det (B_1) - \tfrac{1}{n^2} \det(B)  = (1-\tfrac{2}{n}) \det(B) - \det (B_1) .$$
By assumption $n\geq 2$. Further, $\det(B)$ has sign $(-1)^{j-2}$ if non-zero, and $\det(B_1)$ has sign $(-1)^{j-3}$ if non-zero. Hence, the determinant has sign $(-1)^{j}$ if non-zero.

Finally, assume $z_1=1$ and $z_2=1$. Then $j\geq 4$. The matrix $B$ has at least two blocks. If $B$ has more than two blocks, say $\ell$ blocks $B_1,\dots,B_\ell$, the determinant of 
$\widetilde{J}_H$ agrees with the product of the determinants of $B_i$, $i=2,\dots,\ell-1$, times the determinant of a matrix similar to $\widetilde J_H$ but with only two blocks $B_1$ and $B_\ell$. Since the determinant of $B_i$, $i=2,\dots,\ell-1$, has the desired sign, it is sufficient to show that the determinant of $\widetilde{J}_H$ has sign $(-1)^{j}$ whenever $B$ has exactly two blocks:
{\small $$\widetilde{J}_H= \left(  \begin{array}{ccccc|ccccc}   -1 +\frac{1}{n} &  1 & 0 & \ldots & 0 & 0  & \dots & \dots & 0 & - \frac{1}{n}  \\  
1 & -2  & 1 & \dots &   0  & 0 & \dots & \dots & \dots & 0 \\ 
\vdots & \ddots & \ddots &\ddots & \vdots  & \vdots & \dots & \dots & \dots & \vdots \\ 
0 & \dots & 1 & -2 & 1 &  \vdots & \dots & \dots & \dots & \vdots  \\ 
0 & \dots & 0 & 1 & -2 &  0 & \dots & \dots & \dots & 0 \\  \hline
0 & \dots & \dots & \dots & 0 & -2 & 1 & 0 & \dots & 0 \\ 
\vdots & \dots & \dots & \dots & \vdots & 1 & -2 & 1 & \dots & 0 \\
\vdots & \dots & \dots & \dots & \vdots & \vdots & \ddots & \ddots & \ddots & \vdots \\
0 & \dots & \dots & \dots & 0 & 0 & \ddots & 1 & -2 & 1 \\
- \frac{1}{n} & 0 & \dots & \dots & 0 & 0 & \dots & 0 & 1 & -1 +\frac{1}{n} 
\end{array}\right).$$}

 We will show by induction in $j$ that  the determinant is $(-1)^j \tfrac{n-j}{n}$, which has the desired sign. For the induction basis we need to consider the two smallest cases $j=4,5$, and it is also convenient to check separately the case $j=6$. For $j=4$, the matrix $B$ is 
$$\left(  \begin{array}{cccc}   -1 +\frac{1}{n} &  1 & 0 &  - \frac{1}{n}  \\  
1 & -2  & 0  & 0 \\   0 & 0 & -2 & 1  \\ 
-\frac{1}{n}  & 0 & 1 & -1 + \frac{1}{n}\end{array}\right),   $$
which has  determinant $\tfrac{n-4}{n}$. Similarly, for $j=5,6$, we check that  the determinant is  $-\tfrac{n-5}{n}$ and $\tfrac{n-6}{n}$ respectively.

Assume now that the statement holds for all $4\leq j'\leq j-1$. We will prove the statement for $j'=j$, with $j\geq 7$.
Let $i,i+1$ be the indices of the last column of the first block and the first column of the second block of $B$, respectively.
Since $j\geq 7$ and $i\geq 2$, at least one of the two inequalities hold $i\geq 4$ or $j-i\geq 4$. We assume that $i\geq 4$, meaning the first block of $B$ has at least size $3$, and the other case follows symmetrically. 
The statement will follow by applying the Laplace expansion of the determinant of $\widetilde{J}_H$ along the rows $i-1,i$. For this, we let $\widetilde{J}_{H,\{i-1,i\},\{j_1,j_2 \}}$ be the matrix obtained by removing the $(i-1)$-th and $i$-th rows and the $j_1$-th and $j_2$-th columns of $\widetilde{J}_H$.
Then
\begin{align*}
 \det(\widetilde{J}_H) &=  \det\begin{pmatrix} 1 & -2 \\ 0 & 1 \end{pmatrix} \det\Big(\widetilde{J}_{H,\{i-1,i\},\{i-2,i-1 \}}\Big)  + \det\begin{pmatrix} -2 & 1 \\ 1 & -2 \end{pmatrix} \det\Big(\widetilde{J}_{H,\{i-1,i\},\{i-1,i \}}\Big) \\ &  - \det\begin{pmatrix} 1 & 1 \\ 0 & -2 \end{pmatrix} \det\Big(\widetilde{J}_{H,\{i-1,i\},\{i-2,i \}}\Big).
\end{align*}
The matrix 
$\widetilde{J}_{H,\{i-1,i\},\{i-2,i-1 \}}$ has one zero column, hence the first term is zero. 
By the induction hypothesis, $ \det\Big(\widetilde{J}_{H,\{i-1,i\},\{i-1,i \}}\Big)=(-1)^{j-2} \frac{n-j+2}{n}$. The $(i-2)$-th column of $\widetilde{J}_{H,\{i-1,i\},\{i-2,i \}}$ has only one nonzero entry, equal to one, in the $(i-2)$-th entry. Removing the $(i-2)$-th column and row of this matrix, we obtain a matrix of the form $\widetilde{J}_H$ of size $j-3$. Here we use that $i\geq 4$.
These considerations give:
\begin{align*}
 \det(\widetilde{J}_H) &= 3(-1)^{j-2} \tfrac{n-j+2}{n} + 2(-1)^{j-3} \tfrac{n-j+3}{n} =(-1)^{j-2} ( (3-2)\tfrac{n-j}{n} 
+ 3\tfrac{2}{n} - 2\tfrac{3}{n} =(-1)^j \tfrac{n-j}{n},
\end{align*}
as claimed. This concludes the proof.
\end{proof}

\subsection{Study of the centre manifold} 
\label{sec:centre}

The next step towards the proof of Theorem~\ref{thm:MM} is to complete the study of the stability properties of the steady state $p$ by inspecting its centre manifold.  
Recall the MM system with our choice of parameters given in \eqref{eq:MMgood}, and the two orthogonal vectors spanning the kernel of the Jacobian of this system evaluated at the steady state $p$ (see Proposition \ref{eq:eigenvalues}):
\begin{align*}
v &=(v_0,\ldots,v_n),\quad v_i = n - 2i, \quad i=0,\dots,n,\\
e &=(1,\ldots,1).
\end{align*}
Recall also that the dynamics of \eqref{eq:MMgood} around $p$ is confined to the linear subspace $n+1 = x_0+\dots+x_n$. Let us refer to the system within this 
subspace as the restricted system. Since $v$ satisfies $e\cdot v=0$, the centre 
subspace of the steady state with respect to the restricted system is spanned by $v$.  This implies that the  centre manifold of $p$  is one dimensional and admits a parametrisation (at least locally around $p$) of the form 
\begin{equation}\label{eq:param_centre}
 x_ i(s)= 1 + v_i s + h_i(s), \qquad i=0,\dots,n,
\end{equation}
where 
\begin{equation}\label{eq:h}
 h_0(s)=0,\qquad h_n(s)= -\sum_{\ell=0}^{n-1} h_\ell(s),
 \end{equation}
and all $h_i$ vanish at least as fast as $s^2$ near $0$.
Let $h(s)=(h_0(s),\dots,h_n(s))$.
Since $x_0(s)= 1 +ns$, we have $s(x)=\frac{x_0-1}{n}$ for $x$ on the centre manifold (at least locally around $p$). Although the centre manifold may not be uniquely determined this will not cause a problem. The general theory says that for two different choices of the centre manifold, if the system, and hence the functions $h_i$, are $\mathcal{C}^k$, then the functions $h_i$ for the two choices of the centre manifold agree up to order $k$. Thus we choose a fixed centre manifold and the arguments which follow are independent of that choice. Note that
since the system itself is $\mathcal{C}^\infty$ there exists a centre manifold
of class $\mathcal{C}^k$ for any finite $k$. To justify the calculations 
in the following we just need to choose $k$ sufficiently large.

We investigate the equations of the  centre manifold  by selecting $n-1$ linearly independent vectors $\nu_1,\dots,\nu_{n-1}$ in 
$\langle v, e \rangle ^{\perp}\subseteq \R^{n+1}$.   Using \eqref{eq:param_centre} we obtain the  equations $ \nu_i \cdot x(s) = \nu_i \cdot h(s)$, $i=1,\dots,n-1,$
or equally,  in terms of the variable $x$ (on the centre manifold),
$$\nu_i \cdot x = \nu_i \cdot h(s(x)),\qquad i=1,\dots,n-1.$$
If $f$ denotes the right-hand side of the MM system \eqref{eq:MMgood}, then evaluation at $x=x(\tau)$ and differentiation with respect to time $\tau$ gives
\begin{align*}
 \nu_i \cdot f(x(\tau)) &= \nu_i \cdot h'(s(x(\tau))) \,\frac{d}{d\tau} s(x(\tau))   \nonumber \\
 &=\nu_i \cdot h'(s(x(\tau))) \left( \frac{d}{d\tau} \frac{x_0(\tau)-1}{n}\right) \\
 &=\frac{f_0(x(\tau)) }{n} \,\nu_i \cdot h'(s(x(\tau))),  \nonumber
\end{align*}
or simply,
\begin{equation}\label{eq:fscalar}
\nu_i \cdot f(x)=\frac{f_0(x) }{n} \,\nu_i \cdot h'(s(x)),\qquad i=1,\dots,n-1. 
\end{equation}
(Note that  $x(\tau)$ and $x(s)$ mean two different things.)

Consider the vector field $f$ multiplied by the denominators of $f$ and expressed in terms of the parameterisation of the centre manifold. It takes the form
\begin{align*}
w_0(s) &=  -x_0(s)\left(1+\sum_{\ell=1}^n \beta_{\ell} x_\ell(s) \right)    + x_1(s)\left(1+\sum_{\ell=1}^{n} \alpha_{\ell} x_{\ell-1}(s) \right), \\
w_i(s) &= (x_{i-1}(s)-x_i(s))\left(1+\sum_{\ell=1}^n \beta_{\ell} x_\ell(s) \right) 
+(x_{i+1}(s)-x_i(s))\left(1+\sum_{\ell=1}^{n} \alpha_{\ell} x_{\ell-1}(s) \right), \quad i=1,\dots,n-1,\\
w_n(s) &=   x_{n-1}(s)\left(1+\sum_{\ell=1}^n \beta_{\ell} x_\ell(s) \right)  -x_n(s)\left(1+\sum_{\ell=1}^{n} \alpha_{\ell} x_{\ell-1}(s) \right).
\end{align*}
Let $w(s)=(w_0(s),\dots,w_n(s))$. Then, using \eqref{eq:fscalar} and the definition of $w$, we find
\begin{equation}\label{eq:vanish}
n( \nu_j\cdot w) = (\nu_j \cdot h') w_0, \qquad j=1,\dots,n-1,
\end{equation}
where the dependence on $s$ is suppressed.
This expression implies that the linear combination on the left side vanishes at least at one order higher than $w_0$, since $h'$ vanishes at least at order one. 

Our goal is to show that the first non-zero coefficient in the Taylor expansion of $w_0$  is negative and of order $n+1$ if $n$ is even and of order $n$ if $n$ is odd. If this is so, then it must be that locally around $p$ the flow of the vector field $w(s)$ is attracted towards $p$ as a negative coefficient of $w_0(s)$ implies the absolute value  $|s|$ is decreasing towards zero. Specifically, we prove the following theorem, where $w^{(m)}(s)$ denotes the $m$-th order term in the Taylor expansion of $w$.

\begin{theorem}\label{thm:taylor}
The first non-vanishing term in the Taylor expansion of $w_0$ is negative.
Specifically, the first non-vanishing term is
$$w_0^{(M)}(s)=-  \frac{3\cdot 2^{n+2}\, n}{(n+1)^2(n+2)}s^M,$$
with $M=n+1$ if $n$ is even and $M=n$ if $n$ is odd.
\end{theorem}

In order to prove Theorem~\ref{thm:taylor}, we investigate the order terms $w_i^{(m)}(s)$ of $w_i(s)$ through a series of technical lemmas.  We start by introducing a new quantity. Let
\begin{equation}\label{eq:gamma2}
  \Gamma_n(s) := \sum_{\ell=0}^n \gamma_{n}(\ell) h_\ell(s)=\gamma_n\cdot h(s),
  \end{equation}
and note that the zeroth and first order terms of  $\Gamma_n$ vanish. 
The proof of the next lemma is in Subsection~\ref{proof:gammaalpha}.

\begin{lemma}\label{lem:gammaalpha}
Based on \eqref{eq:param_centre}, the following two relations hold
\begin{align*}
1+\sum_{\ell=1}^n \beta_{\ell} x_\ell(s) &= 
 \alpha_1 n - \alpha_1n s   - \Gamma_n(s).  \\
1+\sum_{\ell=1}^{n} \alpha_{\ell} x_{\ell-1}(s) &= 
\alpha_1 n + \alpha_1 n s - \alpha_1 h_n(s)   - \Gamma_n(s).
\end{align*}
\end{lemma}
 
 Using Lemma~\ref{lem:gammaalpha} and  \eqref{eq:param_centre}, we have the following expressions: \begin{align}\label{eq:w}
w_0(s) &=2\alpha_1 n(n-1)s^2  +\alpha_1(n h_{1}-h_{n}) +\alpha_1(  (2-n)h_{n} + n h_{1}) s -\alpha_1 h_{n}h_{1}+(   2s-h_{1})\Gamma_n,  \nonumber\\
w_i(s) &=-4\,\alpha_1 ns^2 +\,\alpha_1 n(h_{i-1} -2h_i +h_{i+1}  ) +\alpha_1( -nh_{i-1}+n h_{i+1}+2h_{n})s 
  \\ &\quad +\alpha_1 h_{n}(h_{i}-h_{i+1}) +    (-h_{i-1}+2 h_i -h_{i+1})\Gamma_n, \hspace{3cm} \text{for}\quad i=1,\dots,n-1, \nonumber \\ 
w_n(s) &=  2 \alpha_1 n(n-1) s^2  + \alpha_1  ( n h_{n-1} + ( 1-n) h_{n} ) -\,\alpha_1 n (2h_{n}+h_{n-1})s\nonumber  \\ &\quad  +     \alpha_1 h_{n}^{2}+  ( h_n-  h_{n-1}-2s)\Gamma_n, \nonumber
\end{align}
where  the argument $s$ of $h_i$ and $\Gamma_n$ is omitted to ease notation.  The details of the proof of \eqref{eq:w} are not given. A key ingredient in the proof of Theorem~\ref{thm:taylor} is the following function of $s$:
 \begin{align}\label{eq:chi-omega1}
\chi_n(s)=\sum_{j=1}^n j w_j(s)= -\frac{1}{2}v\cdot w(s),
\end{align}
where the second equality follows from the trivial fact that $e\cdot w=0$ and the definition of $v$. 
We further obtain (see Subsection~\ref{proof:chi} for a proof):
 \begin{align}\label{eq:chi-omega2}
 \chi_{n}(s)=(h_n(s)- 2ns  ) \Gamma_n(s),
\end{align}
which implies that terms of $\chi_n$ of order zero, one and two vanish. We also note that $w^{(0)}(s)=w^{(1)}(s)=0$, which trivially follows from \eqref{eq:w}.

\begin{lemma}\label{lem:obs1}
Consider the following statements, where $m\ge 1$ and  $\nu_1,\dots,\nu_{n-1}$ is a basis of 
$\langle v, e \rangle ^{\perp}\subseteq \R^{n+1}$:
\begin{enumerate}[(i)]
\item  $w_0^{(m-1)}(s)=0$.
\item $\nu_j\cdot w^{(m)}(s)=0$ for all $j=1,\ldots,n-1$.
\item $w^{(m)}(s)=K_m v s^m$, that is, $w_i^{(m)}(s)=K_m v_i s^m$ for all  $i=0,\ldots,n$ and some $K_m\in\R$.
\item $\chi_n^{(m)}(s)=-K_m\frac{n(n+1)(n+2)}{6} s^m$ for some $K_m\in\R$.
\end{enumerate}
Then, for $m\geq 1$,  the following  implications hold:
$$ \textrm{(i)} \quad \Rightarrow \quad  \textrm{(ii)} \quad \Rightarrow \quad  \textrm{(iii)} \quad \Rightarrow \quad  \textrm{(iv)}. $$
\end{lemma}

\begin{proof}
 (i) $\Rightarrow$ (ii) follows from \eqref{eq:vanish}. If (ii)  holds, since $e\cdot w=0$, we have  $w^{(m)}\in \langle \nu_1,\dots,\nu_{n-1}, e\rangle^\perp = \langle v \rangle$, which implies (iii).  Finally, assume (iii), then, using  \eqref{eq:chi-omega1},
we have
\begin{align*}
\chi_n^{(m)} &=\sum_{j=1}^n j w_j^{(m)} =  \sum_{j=1}^n j  K_m (n-2j) s^m = K_m s^m \left( n\sum_{j=1}^n j - 2\sum_{j=1}^n j^2\right) =  K_m s^m \Big( \tfrac{n^2(n+1)}{2} -  \tfrac{2n(n+1)(2n+1)}{6}\Big) \\
&= -K_m\tfrac{n(n+1)(n+2)}{6} s^m.
\end{align*}
This completes the proof.
\end{proof}

Let 
\begin{equation}\label{eq:c}
\begin{aligned}
c_0^{(m)}  &=\Big( - 2n(n-1)s^2  -(  (2-n)h_{n} + n h_{1}) s + h_{n}h_{1} -(   2s-h_{1})\tfrac{\Gamma_n(s)}{\alpha_1}\Big)^{(m)}, \\
c_i^{(m)} &=\Big(4 ns^2 -  ( -nh_{i-1}+n h_{i+1}+2h_{n})s -   h_{n}(h_{i}-h_{i+1})
 -    (-h_{i-1}+2 h_i -h_{i+1})\tfrac{\Gamma_n(s)}{\alpha_1}\Big)^{(m)},  \quad \\ 
c_n^{(m)} &= \Big(-   2n(n-1) s^2 + n (2h_{n}+h_{n-1})s  -      h_{n}^{2}-  ( h_n-  h_{n-1}-2s)\tfrac{\Gamma_n(s)}{\alpha_1}\Big)^{(m)},
\end{aligned}
\end{equation}
where $i=1,\ldots,n-1$,
and consider the matrix 
$$ B= \begin{pmatrix}
1 & 2 & 3 & \dots & n-2 & \tfrac{-2}{n} \\[5pt]
0 & 1 & 2 & \dots & n-3 & \tfrac{-3}{n} \\
\vdots & \ddots & \ddots  & \ddots & \vdots & \vdots \\
\vdots & \ddots & \ddots  & 1 & 2 & \tfrac{-(n-2)}{n} \\
\vdots & \ddots & \ddots & \ddots & 1 & \tfrac{-(n-1)}{n} \\ 
0 & \dots & \dots & \dots & 0 & -1
\end{pmatrix} \in \R^{(n-1)\times (n-1)}.
$$

\begin{lemma}\label{lemma:wm=0}
If $w^{(m)}=0$ and $\Gamma_n^{(m-1)}=0$, then
$$ h_i^{(m)} = i h_1^{(m)} + \Big( B \big(  \tfrac{c_3^{(m)}}{n} ,  \dots ,   \tfrac{c_n^{(m)}}{n}  , c_0^{(m)}\big)^t \Big)_{i-1},\qquad i=2,\dots,n. $$
\end{lemma}
\begin{proof}
First of all note that $e\cdot w^{(m)}=0$, and also $(n+1,n,\dots,2,1)\cdot w^{(m)}=0$. Indeed, 
 $(n+1,n,\dots,2,1)=-(0,\dots,n) + (n+1)e$ and hence, by \eqref{eq:chi-omega2},
$$(n+1,n,\dots,2,1)\cdot w(s)^{(m)}= - \chi_n(s)^{(m)}  =  \Big((h_n(s) - 2ns  ) \Gamma_n(s)\Big)^{(m)} =0, $$
by hypothesis.  

The system $w^{(m)}=0$  can be written in matrix form as
\begin{equation}\label{eq:wzero}
 \widehat{A} \big(h_1^{(m)},\dots,h_n^{(m)}\big)^t = c^{(m)}
\end{equation}
with $c^{(m)} = \big(c_0^{(m)},\dots,c_n^{(m)}\big)^t$ and
{\small $$ \widehat{A} = \begin{pmatrix}
n & 0  & \dots & 0   & -1 \\
-2n & n & 0 & \dots & 0 \\
n & -2n & n & \ddots & \vdots \\
\vdots & \ddots & \ddots & \ddots  & 0 \\ 
\vdots & \vdots &   n&-2n& n 
\\
0 & \dots & \dots & n & 1-n
\end{pmatrix} \in \R^{(n+1)\times n}.
$$}%
This matrix has rank $n-1$, since  the rows $2,\dots,n$ are linearly independent. Since $e \widehat{A} = (n+1,n,\dots,2,1)\widehat{A} =0$, also $e\cdot c^{(m)}= (n+1,n,\dots,2,1)\cdot c^{(m)}=0$. Hence, by  deleting the second and the third row of $\widehat{A}$, moving the first row to be the last, the first column to be the last, dividing   the first $n-2$ resulting equations by $n$,  and reorganising the vector $(h_1^{(m)},\dots,h_n^{(m)})$,  system \eqref{eq:wzero} can be rewritten as
\begin{equation}\label{eq:reducedsystem}
A \left(\begin{array}{c} h_2^{(m)} \\  \vdots \\ h_n^{(m)} \\ h_1^{(m)}   \end{array}\right) =
\left(\begin{array}{c}   \tfrac{c_3^{(m)}}{n} \\  \vdots \\   \tfrac{c_n^{(m)}}{n}   \\ c_0^{(m)}   \end{array}\right),
\quad\textrm{with }\quad
A=    \begin{pmatrix}
1 & -2 & 1 & 0 & \dots & 0\\
0  & \ddots & \ddots & \ddots &\vdots   & \vdots \\ 
\vdots & \ddots &   1 &-2& 1  & 0 
\\
\vdots & \ddots & \ddots & 1 & \tfrac{1-n}{n}  &0 \\
 0  & \dots &  \dots & 0   & -1  & n
\end{pmatrix} \in \R^{(n-1)\times n}.
\end{equation}
A straightforward  computation shows that 
$$ B A = \begin{pmatrix}
1 & 0 & \dots & \dots & 0 & -2 \\
0 & 1 & \ddots &\ddots &  \vdots & -3 \\
\vdots & \ddots & \ddots  & \ddots  & \vdots & \vdots \\
\vdots & \ddots & \ddots & 1 & 0 & -n+1 \\
0  & \dots & \dots & 0 & 1 & -n
\end{pmatrix}. $$
Multiplying both sides of  equality \eqref{eq:reducedsystem} by $B$ gives the expression in the statement.
\end{proof}

Let
$$\ell_{[k]}= \ell(\ell-1)\cdots(\ell-k+1)= \frac{\ell!}{(\ell-k)!} = \binom{\ell}{k} k!$$
be the $k$-th descending factorial for $k=0,\ldots,n-1$.  By definition, $\ell_{[k]} =0$ if $k>\ell$.  If $k =0$, then $\ell_{[0]}=1$ for all $\ell\ge 0$. Furthermore, $0_{[0]}=1$ and $0_{[k]}=0$ for $k>0$. 
We introduce the following vectors in $\R^{n+1}$:
 $$e_{[k]}=\big(0_{[k]},\ldots,n_{[k]}\big),\qquad k=0,\dots,n.$$
 In particular,  $e_{[1]}=(0,1,2,\dots,n)$ and $e_{[2]}=(0,0,2,6,\dots,i(i-1),\dots,n(n-1))$.
The proof of the following lemma is given in Subsection~\ref{proof:useful}.

\begin{lemma}\label{lem:useful}
For $k=0,\ldots,n-1$, if $n$ is even, and for $k=0,\ldots,n-2$, if $n$ is odd, we have
$$e_{[k]}\cdot \gamma_n=0.$$
Furthermore,
$$ e_{[n]}\cdot \gamma_n =-\frac{n!}{n+1} \quad \textrm{for }\,n \textrm{ even}, \qquad e_{[n-1]}\cdot \gamma_n =-\frac{2(n-1)!}{n+1}  \quad \textrm{for }\,n \textrm{ odd}. $$
\end{lemma}
 
The above lemma shows that $\gamma_n$ is orthogonal to $e_{[k]}$ for all $k=0,\ldots,n-1$ if $n$ is even, and  for all $k=0,\dots,n-2$ if $n$ odd.

We note also the following identities 
for any non-negative integer $k\ge 0$ and $i\geq k$:
\begin{align}\label{eq:useful2}
 \sum_{\ell=0}^{i} \ell_{[k]}  = \sum_{\ell=k}^i \binom{\ell}{k} k!=\binom{i+1}{k+1}k!=\frac{1}{k+1}i_{[k+1]}+i_{[k]}. 
\end{align}

\begin{lemma}\label{lem:obs3} 
Let $M>0$ and assume that  $\Gamma_n^{(m)}=0$, for all $0\le m\le M$.
Then  we have
\begin{enumerate}[(i)]
\item $w^{(m)}=0$ for $0\le m\le M+1$,
\item  for  $2\le m\le M+1$,
$$h^{(m)}(s)=\left(\sum_{j=1}^m c_{mj} e_{[j]}\right) \!s^m, \quad\text{or equivalently}\quad h_\ell^{(m)}(s)=\left(\sum_{j=1}^m c_{mj} \ell_{[j]}\right)\! s^m,\quad \ell=0,\ldots,n,$$
where $c_{mj}\in\R$, $j=1,\ldots,m$, are constants such that $c_{mm}=(-1)^m\frac{2^m}{m!}$.
\end{enumerate}
\end{lemma}

\begin{proof} (i) 
The proof is by induction on $m$. By construction $w^{(0)}=0$ for  $m=0$.  
Now assume that  $w^{(m)}=0$ for an index satisfying $0\leq m \leq M$. We show that the statement holds for $m+1$.
Since $w^{(m)}=0$, then from Lemma \ref{lem:obs1}, $w^{(m+1)}=K_{m+1} v s^{m+1}$ and $\chi_n^{(m+1)}=-K_{m+1}\frac{n(n+1)(n+2)}{6} s^{m+1}$. 
On the other hand, by assumption, $\Gamma_n^{(j)}=0$ for $0\le j\le M$, in particular for $0\leq j \leq m$, and combining this with equation \eqref{eq:chi-omega2} yields
$$\chi_{n}^{(m+1)}(s)=  \sum_{j=2}^{m-1}h_n^{(j)}(s)\Gamma_n^{(m+1-j)}(s)- 2ns \, \Gamma_n^{(m)}(s)=0.$$
Hence $K_{m+1}=0$, and thus $w^{(m+1)}=0$. This shows (i).

(ii) is proven by induction in $m$. Consider $m=2$. The form of $h_i^{(2)}$ follows from Lemma~\ref{lemma:wm=0}, using that $c_0^{(2)} = - 2n(n-1)s^2$, $c_i^{(2)} =4 ns^2$, $i\neq 0,n$ and $c_n^{(2)}= -  2 n(n-1) s^2$, see \eqref{eq:c}. Indeed,   $h_i^{(2)}$ equals $i h_1^{(2)}$ plus  the $(i-1)$-th component of the matrix product  $  B (\tfrac{c_3^{(2)}}{n},\dots,\tfrac{c_n^{(2)}}{n},c_0^{(2)})^{t}$, which is
\begin{multline*}
\tfrac{1}{n} c_{i+1}^{(2)} + \tfrac{2}{n}c_{i+2}^{(2)} + \dots + \tfrac{(n-i)}{n} c_n^{(2)} -  \tfrac{i}{n} c_0^{(2)} =
s^2 \big(2(n-1) i +  \sum_{k=1}^{n-i-1} 4k  -  2(n-1)(n-i) \big)=2 i (i-1) s^2.
\end{multline*}
In vector notation, this is 
 $h^{(2)}=h_1^{(2)} e_{[1]}+ 2 s^2 e_{[2]}$. Since $h_1^{(2)}$ is of the form $c_{21} s^2$ for some $c_{21}\in\R$ we have
$$h^{(2)}=\left(c_{21} e_{[1]}+c_{22}e_{[2]}\right)\!s^2,\quad c_{22}=2,$$
and the claim is true for $m=2$. 

Now assume the claim is true for all $m'$  such that $2\le m' \le m\le M$ and consider $m'=m+1$. 
We start by describing $h^{(m+1)}$. Since $w^{(m+1)}=0$ by (i) and $\Gamma_n^{(m)}=0$, we might apply Lemma~\ref{lemma:wm=0}. Since $\Gamma_n^{(m')}=0$ for $0\le m'\le m$ by assumption, and  $m\ge 2$, it holds that 
\begin{align*}
c_0^{(m+1)}  &=  -\Big(  (2-n)h_{n}^{(m)} + n h_{1}^{(m)}\Big) s + \sum_{j=2}^{m-1}h_1^{(j)}h_n^{(m+1-j)}, \\
c_i^{(m+1)} &=   -  \Big( -nh_{i-1}^{(m)}+n h_{i+1}^{(m)}+2h_{n}^{(m)}\Big)s -   \sum_{j=2}^{m-1}\big(h_{i}^{(j)}-h_{i+1}^{(j)}\big) h_n^{(m+1-j)} ,\qquad i=1,\dots,n-1,\\ 
c_n^{(m+1)} &=    n \Big(2h_{n}^{(m)}+h_{n-1}^{(m)}\Big)s  -    \sum_{j=2}^{m-1}h_n^{(j)}h_n^{(m+1-j)}.
\end{align*}
Then, by Lemma~\ref{lemma:wm=0}, $h_i^{(m+1)}$ equals $i h_1^{(m+1)}$ plus  the $(i-1)$-th component of the matrix product 
$  B \big(\tfrac{c_3^{(m+1)}}{n},\dots,\tfrac{c_n^{(m+1)}}{n},c_0^{(m+1)}\big)^{t}$. That is, for $i=2,\dots,n$, 
$$
h_i^{(m+1)} = i h_1^{(m+1)}+ \tfrac{1}{n} c_{i+1}^{(m+1)} + \tfrac{2}{n}c_{i+2}^{(m+1)} + \dots + \tfrac{(n-i)}{n} c_n^{(m+1)} -  \tfrac{i}{n} c_0^{(m+1)} = i h_1^{(m+1)}+ D_1 s + D_2,$$
where
\begin{align*}
D_1 &= \left(\sum_{j=i+1}^{n-1} (j-i) \big(  h_{j-1}^{(m)} -  h_{j+1}^{(m)}- \tfrac{2}{n} h_{n}^{(m)}\big) \right)+ (n-i) \big(2h_{n}^{(m)}+h_{n-1}^{(m)}\big)  +  \tfrac{i}{n} \big(  (2-n)h_{n}^{(m)} + n h_{1}^{(m)}\big)  \\
&= i h_{1}^{(m)} + h_n^{(m)} \Big(\tfrac{-2}{n}\Big( \sum_{j=i+1}^{n-1} (j-i)\Big)+ 2(n-i)   +  \tfrac{i}{n}  (2-n)  \Big) + (n-i) h_{n-1}^{(m)}+  \sum_{j=i+1}^{n-1} (j-i) \big(  h_{j-1}^{(m)} -  h_{j+1}^{(m)}\big) \\
&=  i h_{1}^{(m)} + h_n^{(m)}  \Big(2+\tfrac{i(1-i)}{n}\Big) + h_i^{(m)} + 2 \sum_{\ell=i+1}^{n-1} h_\ell^{(m)}
 = i h_{1}^{(m)} + \tfrac{i(1-i)}{n} h_n^{(m)}   + h_i^{(m)} + 2 \sum_{\ell=i+1}^{n} h_\ell^{(m)},
\end{align*}
and 
\begin{align*}
D_2 &= - \sum_{\ell=i+1}^{n-1} \tfrac{\ell-i}{n} \left( \sum_{j=2}^{m-1}(h_{\ell}^{(j)}-h_{\ell+1}^{(j)}) h_n^{(m+1-j)} \right)  - \tfrac{(n-i)}{n}    \sum_{j=2}^{m-1}h_n^{(j)}h_n^{(m+1-j)} - \tfrac{i}{n} \left( \sum_{j=2}^{m-1}h_1^{(j)}h_n^{(m+1-j)}\right) \\
&=-  \sum_{j=2}^{m-1}\left( \left( \sum_{\ell=i+1}^{n-1} \tfrac{\ell-i}{n}  (h_{\ell}^{(j)}-h_{\ell+1}^{(j)}) \right) +\tfrac{(n-i)}{n}   h_n^{(j)}  + \tfrac{i}{n}  h_1^{(j)}\right)h_n^{(m+1-j)}\\ &  =  -\tfrac{1}{n} \sum_{j=2}^{m-1} \Big( \sum_{\ell=i+1}^{n} h_{\ell}^{(j)}   +i h_1^{(j)}\Big)h_n^{(m+1-j)}.
\end{align*}

Let $g^{(q)}$ be defined component-wise as
$g_i^{(q)}  =  \sum_{\ell=i+1}^{n} h_{\ell}^{(q)} $. Then, the above result might be given in  vector notation as:
\begin{align}
h^{(m+1)}&=h_1^{(m+1)}e_{[1]}+\left(h_1^{(m)}e_{[1]} -\tfrac{1}{n}h_n^{(m)}e_{[2]} +h^{(m)}+2g^{(m)}\right)\!s  \nonumber
\\ & \qquad -\tfrac{1}{n}\sum_{j=2}^{m-1} \left(g^{(j)} + h_1^{(j)}e_{[1]}\right) h_n^{(m+1-j)}. \label{eq:hvector}
\end{align}
Using $\sum_{\ell=0}^{n} h_\ell^{(q)}=0$ \eqref{eq:h}, equation \eqref{eq:useful2},
and  the induction hypothesis, we note that 
\begin{align*}
g_i^{(q)} &= - \sum_{\ell=0}^i h_\ell ^{(q)}=-  \sum_{\ell=0}^i \left(\sum_{j=1}^{q} c_{qj} \ell_{[j]}\right) \!s^{q}= - \left( \sum_{j=1}^q  c_{qj}\sum_{\ell=0}^i \ell_{[j]}\right) \!s^q \\
&=- \sum_{j=1}^q  c_{qj} \left(  \tfrac{i_{[j+1]}}{j+1}+i_{[j]} \right) \!s^q=- \sum_{j=1}^{q+1}   \left(  \tfrac{c_{q(j-1)}}{j}+c_{qj} \right) \! i_{[j]}s^q, \qquad 2\leq q\leq m,
\end{align*}
where $c_{q0}=c_{q(q+1)}=0$. Consequently, 
$$g^{(q)}=- \sum_{j=1}^{q+1}   \left(  \frac{c_{q(j-1)}}{j}+c_{qj} \right) \! e_{[j]}s^q,\qquad 2\leq q\leq m.$$
Note that $h_j^{(q)}(s)= z_{qj}s^q$ for some $z_{qj}\in \R$. 
Using this expression and the induction hypothesis on \eqref{eq:hvector}, we obtain
\begin{align*}
h^{(m+1)}&=
z_{(m+1)1} s^{m+1} e_{[1]}+\left(z_{m1}s^{m} e_{[1]} -\tfrac{z_{mn} }{n} s^m e_{[2]} +  \sum_{j=1}^m c_{mj} e_{[j]} \, s^m - 2\sum_{j=1}^{m+1}   \left(  \frac{c_{m(j-1)}}{j}+c_{mj} \right) \! e_{[j]}s^m 
\right)\!s \\
&\quad -\tfrac{1}{n}\sum_{j=2}^{m-1} \left( \sum_{k=1}^{j+1} -  \left(  \frac{c_{j(k-1)}}{j}+c_{jk} \right) \! e_{[k]}s^j    + z_{j1}s^je_{[1]}\right) z_{(m+1-j)n} s^{m+1-j}.
\end{align*}
 This expression shows that $h^{(m)}$  takes the form stated in the lemma. The coefficient of $s^{m+1} e_{[m+1]}$ is, by the induction hypothesis, 
$$c_{(m+1)(m+1)}=-2\left( \frac{c_{mm}}{m+1}+c_{m,m+1} \right)=\frac{-2 c_{mm}}{m+1}=\frac{(-1)^{m+1}2^{m+1}}{(m+1)!},$$
as required. 
\end{proof}

\begin{lemma}\label{lemma:gamma}
$\Gamma_n^{(m)}=0$ for all $m=0,\dots,n-1$ if $n$ even and for all $m=0,\dots,n-2$ if $n$ odd.
\end{lemma}
\begin{proof}
Let $M=n-1$ if $n$ is even and $M=n-2$ if $n$ is odd. The proof is by induction in $m$. By construction $\Gamma_n^{(0)}=0$. 
Assume  $\Gamma_n^{(m')}=0$ for all $0\leq m' \leq m<M$, and consider $m+1$.  By Lemma~\ref{lem:obs3}(ii) and the induction hypothesis,  the vector of coefficients of  $s^{m+1}$ in $h^{(m+1)}$  lives in the vector space spanned by the vectors $e_{[1]},\dots,e_{[m+1]}$. 
Now,  using \eqref{eq:gamma2} and Lemma~\ref{lem:useful}, we have
$$\Gamma_n^{(m+1)}=\gamma_n\cdot h^{(m+1)}=0,$$
since $m+1\leq M$.
\end{proof}

We are now in a situation where we can prove  Theorem~\ref{thm:taylor}.  

\begin{proof}[Proof of Theorem \ref{thm:taylor}] Let  $M=n+1$ if $n$ is even and $M=n$ if $n$ is odd.
By Lemma~\ref{lemma:gamma}, $\Gamma_n^{(m)}=0$ for all $0\leq m \leq M-2$ and hence by Lemma~\ref{lem:obs3}(i), $w^{(m)}=0$ (and in particular $w_0^{(m)}=0$), for all $1\leq m \leq M-1$.

By Lemma~\ref{lem:obs1}, we have $w_0^{(M)} = K_{M} n  s^{M}$ and 
$$\chi_n^{(M)}(s)=-K_{M}\frac{n(n+1)(n+2)}{6} s^{M}.$$
Furthermore, by Lemma~\ref{lem:obs3}(ii),  the vector of coefficients of  $s^{M-1}$ in $h^{(M-1)}$  
lives in the vector space spanned by the vectors $e_{[1]},\dots,e_{[M-1]}$ and the coefficient of $e_{[M-1]}$ is $c_{M-1,M-1}$. 
 Using   \eqref{eq:gamma2} and Lemma~\ref{lem:useful}, we have
$$\Gamma_n^{(M-1)}(s)=\gamma_n\cdot h^{(M-1)}(s)=c_{M-1,M-1}\, (\gamma_n\cdot e_{[M-1]})  s^{M-1}.$$
By  \eqref{eq:chi-omega2}, we have
$$
 \chi_{n}^{(M)}(s)=   \sum_{j=2}^{M-1}h_n^{(j)}(s)\Gamma_n^{(M-j)}(s)   -  2n s  \, \Gamma_n^{(M-1)}(s) = - 2 n\,  c_{M-1,M-1}\, (\gamma_n\cdot e_{[M-1]})  s^{M},
$$
where we use that the first summand vanishes. This implies that
$$ -K_{M}\frac{n(n+1)(n+2)}{6} =  - 2 n\,  c_{M-1,M-1}\, (\gamma_n\cdot e_{[M-1]}),  $$
and hence  the first non-zero coefficient of the Taylor expansion of $w_0$ is $n$ times
$$K_{M} =   \frac{12 c_{M-1,M-1}\, (\gamma_n\cdot e_{[M-1]})} {(n+1)(n+2)}.   $$
If $n$ is even, $M-1=n$, $(-1)^n=1$ , and we have by  Lemma~\ref{lem:useful} and  Lemma~\ref{lem:obs3}(ii) that,
$$ K_{n+1} =  -  \frac{12 \cdot  2^n\,   n!} {n!(n+1)^2(n+2)}=  -  \frac{3\cdot  2^{n+2}  } {(n+1)^2(n+2)}. $$
Similarly, if $n$ is odd, then $M-1=n-1$ and
$$ K_n =-  \frac{12\cdot  2^{n-1} \cdot 2\cdot (n-1)! } {(n-1)! (n+1)^2(n+2)}=-  \frac{3\cdot 2^{n+2}}{(n+1)^2(n+2)}. $$
The statement now follows from $w_0^{(M)} = K_{M} n  s^{M}$.
This concludes the proof of Theorem~\ref{thm:taylor}.
\end{proof}

These computations are supported by computations in Maple for up to $n=15$.

\subsection{Concluding the argument}
\label{sec:conclude}

We have shown that the MM system admits a steady state of multiplicity $2\lfloor\tfrac{n}{2}\rfloor+1$ with $n-1$ negative eigenvalues and such that the first non-zero term of the Taylor expansion of $w_0(s)$ is negative. For brevity, as in the previous section, we denote this multiplicity by $M$. Next we will perturb the system so that the multiple steady state splits into $M$ steady states of multiplicity one. 

To do this note first that given any polynomial of degree $n+1$ whose leading term is $1$ and whose constant term is $-1$ it is possible  to choose the coefficients $\alpha_i$ and $\beta_i$ in (\ref{eq:thepoly}) so as to reproduce the given polynomial. Moreover this can be done in such a way that if the coefficients in the polynomial are varied the coefficients  $\alpha_i$ and $\beta_i$ depend smoothly on the coefficients of the polynomial. (In what follows the term 'smooth'
is used to mean $\mathcal{C}^\infty$.)
Next we introduce a specific family of polynomials depending on a parameter
$\mu$ by replacing the function $(u-1)^{2k+1}$ used previously by 
$$\left(u-1\right)\prod_{\ell=1}^k\big(u-(1+\ell\mu)\big)\big(u-(1+\ell\mu)^{-1}\big).$$
This works for $n$ even. We construct a similar family of polynomials if $n$ is odd.
 We can then choose  coefficients $\alpha_i(\mu)$ and $\beta_i(\mu)$  depending smoothly on $\mu$
to reproduce this family. After that we can choose parameters of the MM system 
depending smoothly on $\mu$ so as to give rise to these coefficients 
$\alpha_i$ and $\beta_i$. In the end we have a family of MM systems
depending smoothly on the parameter $\mu$. 

For $\mu=0$, we have the multiple  steady state which has been studied before. For convenience it will be referred to as the bifurcation point. As $\mu$ is varied the root of the polynomial  at $u=1$, which corresponds to the bifurcation point, splits into simple  roots. As is common in bifurcation theory we introduce a suspended system by adjoining the equation $\mu'=0$ to the given evolution equations. This increases the dimension of the system by one. The bifurcation point is also a steady state of the suspended system. Its centre manifold as a solution of that system is of
dimension two and is foliated by invariant curves of constant $\mu$. The  invariant curve with $\mu=0$ is the centre manifold of the bifurcation  point with respect to the MM system studied previously.
 The invariant curve for a non-zero value of $\mu$ will be referred to as the \lq perturbed centre manifold\rq\
although it should be emphasised that it itself is not the centre manifold of anything. It is a general property of a centre manifold of any steady state that all other steady states sufficiently close to the original one lie on that centre manifold. Thus for $\mu$ small all steady states of the  
parameter-dependent MM  system for a fixed value of $\mu$ close to the bifurcation point lie on the perturbed  centre manifold.

The non-zero roots $u$ of the polynomial are simple. This suggests that the  corresponding steady states of the MM system might be  hyperbolic, but this is not obvious. It will now be proved that it is in fact  true. The restriction of the system to the two-dimensional centre manifold of the bifurcation point with respect to the suspended system can be thought of as a one-dimensional dynamical system depending on the  parameter $\mu$. Let us write it in the form $s'(\tau,\mu)=q(s,\mu)$, where $q(s,0)$ is the restriction of the vector field $f$ to the centre manifold parametrised with $s$, 
and $'$ denotes derivative with respect to $\tau$. The  function $q$ has one zero for $\mu=0$ and $M$ zeroes for $\mu\ne 0$. The aim  is to show that $q'$ is non-vanishing at each zero of $q$ for $\mu\ne 0$ and  $\mu$ sufficiently small. This will be proved by contradiction. 

If the statement is false, then there is a sequence $\mu_j$ such that $\mu_j\to 0$ as $j\to\infty$, and such 
that the function $q(s,\mu_j)$ has at least one degenerate root. 
In particular $q(s,\mu_j)$ has at least $M+1$ roots counted with multiplicity. 
Then $q'$ has at least $M$ (distinct) roots, all of which must lie between the smallest and the largest positive roots of $q$.
Continuing in this fashion, it can be concluded that $q^{(M)}(s,\mu_j)$ has at least one root, which is between the smallest and the largest root of $q(s,\mu_j)$. Hence by continuity it follows that $q^{(M)}(0,0)=0$. This
contradicts Theorem \ref{thm:taylor}, saying that $q^{(M)}(0,0)\not=0$,  and so in reality $q'$ does not vanish at any zero of $q$  for $\mu\ne 0$. Note that it does not matter whether we argue about the vector field $w$ or $f$ in Theorem \ref{thm:taylor}. Since $w=rf$ with $r(s)$ an arbitrarily often differentiable positive function with $r(0)=1$, then $f(s)\approx -a s^M$  for some $a>0$ near $s=0$.

It follows from this argument that for any $\mu\ne 0$ the steady states of  the restriction of the dynamical system to the perturbed centre manifold are hyperbolic. Since the other eigenvalues of the linearisation at the bifurcation are negative, it follows by continuity that at nearby points with $\mu\ne 0$, all  eigenvalues have non-zero real parts. More information can be obtained by  considering the sign of the function $q'$ at those points where it is not zero. At each steady state the sign changes as no zeros are degenerate. The  sign can be determined by continuity since in Theorem \ref{thm:taylor} we have established the sign in the case $\mu=0$. Namely, since $M$ is odd and $q^{(M)}(0,0)<0$, locally around zero, $q(s,0)$ is positive for $s<0$ and negative for $s>0$. 
 It follows that within the perturbed centre manifold sinks and sources alternate and the outermost steady states are sinks. Hence if these points, considered as steady states of the full MM system, are   ordered in a suitable way, sinks alternate with saddle points whose unstable manifolds are one-dimensional and the outermost ones are sinks. This completes the proof of Theorem \ref{thm:MM}.

\color{black}

\section{Discussion}

Perhaps it is of relevance to point out differences and similarities to previous work. The proof of Theorem \ref{thm:full} is substantially different from the proof given in \cite{hell15a} for the case $n=2$. Generalising the proof of \cite{hell15a} seems non-trivial and impracticable already for $n=4$, the first case that would have given results establishing the existence of more than the two stable steady states known for $n=2$.

The present proof revolves around several key points. First of all, a reduction of the system using time scale separation  that produces a one-dimensional centre manifold is performed. A similar procedure has been applied to investigate the stability of steady states in a dynamical system modelling a different biological situation, the Calvin cycle   \cite{calvin2017}. Secondly, we obtain the desired number of steady states from a single bifurcation point. Since the centre manifold is one-dimensional, this allows us to conclude the stability of the steady states as we unfold the bifurcation point by parameter perturbation. 

It is natural to wonder to what extent these techniques could be developed into a general procedure or statement, as unlimited multistationarity has been established for other families of systems \cite{feliu:unlimited,feng:allosteric}. For the specific reduction of the system to be applicable, we need the transverse  eigenvalues to have negative real parts. This is true not only for this case but in general for reduction of systems by (so-called) removal of non-interacting species \cite{feliu:intermediates,Fel_elim,saez_reduction,walcher19}. So this part of the proof could potentially be generalised. However, it seems non-trivial to devise techniques to guarantee and analyse a one-dimensional centre manifold for which a bifurcation point exists.

As mentioned in the introduction, the steady states whose stability has been
investigated here are not the most general steady states of the multiple
futile cycle, in the sense that there might exist more than $2\lfloor \tfrac{n}{2}\rfloor +1$ for some parameter choices. It was proven in \cite{{Wang:2008dc}} that $2n-1$ is an upper bound of the number of positive steady states.
It was conjectured in \cite{flockerzi14} that this bound is sharp and the
conjecture was proved in the case $n=3$. For $n\ge 3$ the Michaelis-Menten
system does not have so many steady states and so a reduction to that system
cannot be used to prove the conjecture. Perhaps there exists some other
rescaling leading to a different limiting system that admits $2n-1$  steady states, which can  be lifted to the multiple futile cycle. If this were true, then perhaps the stability 
analysis of the present paper could also be extended to that case.

The strategy of establishing the occurrence of dynamical features of a reaction network by using a reduction in the sense of GSPT is not limited to the case of steady states and their stability. It has been used to prove the existence of periodic solutions of the MAPK cascade, a phosphorylation system  more complicated than the multiple futile cycle \cite{hell16}. In the set-up of GSPT explained above, and under the assumption that there are no purely imaginary transverse eigenvalues, it is sometimes possible to show that the existence of a periodic solution of the limiting system implies the  existence of a periodic solution of the original system. A sufficient  condition for this is that the periodic solution of the limiting  system is hyperbolic, that is,  no eigenvalue of the Poincar\'e map has modulus one. This strategy has also been suggested in the context of reaction networks \cite{banajiOscillation}. If in addition that solution is stable (which in this case means that the modulus of each eigenvalue of the Poincar\'e mapping is less than one), then the solution of the original system is also stable.   In \cite{hell16} the strategy could not be employed since hyperbolicity could not be proved. Instead an alternative strategy was used which might be more widely applicable. The existence of periodic solutions is often proved by showing that there is a Hopf bifurcation. It turns out that a Hopf bifurcation in the limiting system implies the presence of a Hopf  bifurcation in the original system and hence the existence of  periodic solutions of the original system, without any hyperbolicity condition being necessary.

\section{Proofs of auxiliary lemmas}

\subsection{Proof of Lemma~\ref{lemma:Pmatrix}}\label{sec:Pmatrix}
Recall the statement of the lemma: If $M=(m_{ij})$ is a square matrix such that $m_{ij}=a_{ij}+b_i$ where $a_{ij}=0$ for all $i\ne j$ and all $a_{ii}$ and $b_i$ are positive, then all eigenvalues of $M$ have positive real part.

To prove this result, consider, for $n\geq 1$,  the following $n\times n$ matrix:
$$ A_n = \begin{pmatrix}
a_1 + b_1 & b_1 & \dots & b_1 \\ 
b_2 & a_2+b_2 &\dots & b_2 \\
\vdots & \vdots & \ddots &  \vdots \\
b_n & b_n &  \dots & a_n+b_n
\end{pmatrix},$$
where $a_i,b_i>0$. We want to show that all eigenvalues of $A_n$ have positive real part.  Given a matrix $B\in \R^{n\times n}$ and two sets $I,J\subseteq\{ 1,\dots,n\}$ of cardinality $r\leq n$, we denote by  $B_{I,J}\in \R^{r\times r}$ the submatrix of $B$ consisting of the rows indexed by $I$ and the columns indexed by $J$. Then
\begin{itemize}
\item
A matrix $B\in \R^{n\times n}$ is a P-matrix if  all principal minors, $\det(B_{I,I})$ for $I\subseteq\{ 1,\dots,n\}$ are positive. 
\item A matrix $B\in \R^{n\times n}$ is sign-symmetric if for every pair sets $I,J\subseteq\{ 1,\dots,n\}$ of cardinality $r\leq n$,
it holds that
$$\sign(\det(B_{I,J}))\sign(\det(B_{J,I}))\geq 0.$$
\end{itemize}

It follows from \cite{Carlson:1974uw}, that if $B$ is a P-matrix and sign-symmetric, then all eigenvalues of $B$ have positive real part.  Therefore, it is enough to show that $A_n$ is a $P$-matrix and sign-symmetric.  

We first show that $A_n$ is a P-matrix for all $n\geq 1$.
 Because non-maximal principal minors of $A_n$ are of the same form as $A_n$ but of smaller size, it is sufficient to prove that $\det(A_n)>0$ for all $n\ge 1$.  We first subtract the last column of $A_n$ from all other columns and obtain a new matrix $A'_n$:
$$ A'_n = \begin{pmatrix}
a_1  & 0 & 0 & \dots & b_1 \\ 
0 & a_2 & 0 & \dots & b_2 \\
\vdots & \vdots  & \ddots & \vdots &  \vdots \\
0 & 0 & \dots & a_{n-1} & b_{n-1} \\
-a_n & -a_n &  -a_n &\dots & a_n+b_n
\end{pmatrix}.$$
Let $r_i$ denote the $i$-th row of $A_n'$. We replace the last row of $A'_n$,  $r_n$, by the linear combination 
$$ (a_n/a_1)r_1 + (a_n/a_2)r_2 + \dots +(a_n/a_{n-1})r_{n-1}+r_n$$
and obtain the matrix:
$$ A''_n = \begin{pmatrix}
a_1  & 0 & 0 & \dots & b_1 \\ 
0 & a_2 & 0 & \dots & b_2 \\
\vdots & \vdots  & \ddots & \vdots &  \vdots \\
0 & 0 & \dots & a_{n-1} & b_{n-1} \\
0 & 0 &  0 &\dots & a_n + \sum_{i=1}^{n} (a_nb_i/a_i)
\end{pmatrix}.$$
The matrix $ A''_n$ is upper-triangular matrix and the diagonal entries are positive. Therefore, the determinant of $A''_n$ is  positive, and as a consequence $\det(A_n)$ is also positive. This shows that $A_n$ is a P-matrix.

To see that $A_n$  is sign-symmetric  for all $n\geq 1$, we will show that
\begin{equation}\label{eq:signs}
\sign(\det(A_n)_{I,J})= \sign(\det(A_n)_{J,I}),
\end{equation}
for all $I,J\subseteq\{ 1,\dots,n\}$ of cardinality $r$. Clearly, we only need to check the case $I\neq J$.

If $I\cap J$ contains strictly less than $r-1$ elements, then 
$\det((A_n)_{I,J}) = \det((A_n)_{J,I})=0$. Indeed, assume there are two indices $i,j\in I$ that are not in $J$. Then $(A_n)_{I,J}$ has  rows $(b_i\quad b_i\quad \cdots\quad b_i)$ and $(b_j\quad b_j\quad \cdots \quad b_j)$ and hence the determinant is zero. 
By choosing two indices $i,j\in J$ that are not in $I$ we argue that $\det((A_n)_{J,I})=0$ as well.

Therefore, we only need to check \eqref{eq:signs} when $I,J$ differ in only one index.
Let $I\cap J=\{ \ell_1,\dots,\ell_{r-1}\}$ with $\ell_1< \ldots< \ell_{r-1}$ and $i,j,s,k$ such that 
\begin{align*}
I & =  \{ \ell_1,\dots,\ell_{k}, i , \ell_{k+1},\dots,\ell_{r-1}\}, \\
J & = \{ \ell_1,\dots,\ell_{s}, j , \ell_{s+1},\dots,\ell_{r-1}\},
\end{align*}
where $\ell_{k}< i < \ell_{k+1}$ and $\ell_{s}< j < \ell_{s+1}$.
The matrix $ (A_n)_{I,J} $ might be constructed by taking a matrix of type $A_{r-1}$, built from the data $a_{\ell_1},\dots,
a_{\ell_{r-1}}$ and $b_{\ell_{1}},\dots,b_{\ell_{r-1}}$, and adding a row $(b_i \ \dots \ b_i)$ after the $k$-th row and a column with entries  $b_{\ell_1},\dots,b_{\ell_{k}},b_i,b_{\ell_{k+1}},\dots,b_{\ell_{r-1}}$ after the $s$-th column.
By applying the permutation that sends the $k$-th row to the first row and the $s$-th column to the first column, we conclude that 
$\det(A_n)_{I,J} $ agrees with $(-1)^{k+s}$ times the determinant of a matrix of the form 
$$ B_r = \begin{pmatrix}
\beta_1 & \beta_1 & \dots & \beta_1 \\ 
\beta_2 & \alpha_2+\beta_2  & \dots & \beta_2 \\
\vdots & \vdots &  \ddots & \vdots\\
\beta_{r} &\beta_{r} &  \dots & \alpha_{r}+\beta_{r}  
\end{pmatrix}$$
with $(\beta_1,\dots, \beta_r)=(b_i,b_{\ell_1},\dots,b_{\ell_{r-1}})$ and $(\alpha_2,\dots,\alpha_{r-1})=  (a_{\ell_1},\dots,a_{\ell_{r-1}})$.

Similarly, the matrix $ (A_n)_{J,I} $ might  be constructed by taking the matrix  $A_{r-1}$ as above, and adding a row $(b_j \ \dots \ b_j)$ after the $s$-th row and a column with entries  $b_{\ell_1},\dots,b_{\ell_{s}},b_j,b_{\ell_{s+1}},\dots,b_{\ell_{r-1}}$ after the $k$-th column.
Therefore, 
$\det(A_n)_{J,I} $ agrees with $(-1)^{k+s}$ times the determinant of a matrix of the form $B_r$ above, now with 
 $\beta_1=b_j$.

Therefore, it is enough to show that the sign of $\det(B_r)$ does not depend on the values of $\alpha_i>0,\beta_i>0$.
We proceed similarly to the argument given for the first statement. 
We subtract the first column of $B_r$ to all other columns and obtain a new matrix $B'_r$ equal to:
$$ B_r' = \begin{pmatrix}
\beta_1 & 0 & \dots & 0 \\ 
\beta_2 & \alpha_2   & \dots & 0 \\
\vdots & \vdots &  \ddots & \vdots\\
\beta_{r-1} & 0 &  \dots & \alpha_{r-1} 
\end{pmatrix}.$$
The determinant of $B_r'$ is $\beta_1\alpha_2\cdot \dots \cdot \alpha_{r-1}$. Therefore the sign of  $\det(B_r')$, and hence of  $\det(B_r)$, is  $+1$ and  is independent of $\alpha_i,\beta_i$, as desired.
This concludes the proof.

\medskip

\medskip
\noindent
\subsection{Proof of Lemma \ref{lem:gammaalpha}}\label{proof:gammaalpha}

\begin{proof}  
Recall that $\beta_i = \alpha_1 - \gamma_n(i) -\tfrac{1}{n+1}$ (Proposition \ref{prop:good}), that  $\gamma_n(n)=-\tfrac{1}{n+1}$ and that $\sum_{\ell=1}^n \ell  \gamma(\ell)=0$ (Lemma~\ref{lem:useful}). 
By definition of $v_\ell$ and Proposition~\ref{prop:good}, we have
\begin{align*}
 \sum_{\ell=1}^{n}  \beta_\ell v_\ell   &=  \sum_{\ell=1}^{n} \beta_\ell (n-2\ell) = 
n  \sum_{\ell=1}^{n}  \beta_\ell - 2   \sum_{\ell=1}^{n}  \ell \, \alpha_1 + 
 2   \sum_{\ell=1}^{n} \ell   \gamma_{n}(\ell)   + 
   \tfrac{2}{n+1} \sum_{\ell=1}^{n} \ell     \\ 
&= n( \alpha_1 n -1) -\alpha_1 n(n+1) + n  =-\alpha_1 n.
\end{align*}
Then by definition of $\Gamma_n(s)$ and using $\sum_{\ell=1}^n h_\ell =0$ \eqref{eq:h}, we have
\begin{align*}
 \sum_{\ell=1}^{n} \beta_\ell h_\ell(s) = (\alpha_1 - \tfrac{1}{n+1}) \sum_{\ell = 1}^{n} h_\ell(s) - \sum_{\ell=1}^{n} 
\gamma_{n}(\ell) h_\ell(s) =  - \Gamma_n(s).
\end{align*}
Using this information, Proposition~\ref{prop:good} and that $h_0=0$, we  compute the two  factors:
\begin{align*}
1+\sum_{\ell=1}^n \beta_{\ell} x_\ell(s) &= 
1+ \sum_{\ell=1}^n \beta_{\ell} + \sum_{\ell=1}^n  \beta_\ell v_\ell s + \sum_{\ell=1}^n  \beta_\ell h_\ell (s)
  =
 \alpha_1 n - \alpha_1n s   - \Gamma_n(s).  \\
1+\sum_{\ell=1}^{n} \alpha_{\ell} x_{\ell-1}(s) &= 
1+ \sum_{\ell=1}^{n} \alpha_{\ell} +\sum_{\ell=1}^{n} \alpha_{\ell} v_{\ell-1} s +\sum_{\ell=1}^{n} \alpha_{\ell}h_{\ell-1} (s)\\ & = \alpha_1 n +\sum_{\ell=1}^{n-1} \beta_{\ell} v_{\ell} s+ \alpha_1 v_0 s +\sum_{\ell=1}^{n-1} \beta_{\ell}h_{\ell} (s)+ \alpha_1 h_0(s) \\ & 
 = \alpha_1 n - \alpha_1 n s - (\alpha_1 (-n))s + \alpha_1 n s -\alpha_1h_n(s) - \Gamma_n(s)
\\ &  =\alpha_1 n + \alpha_1n s -  \alpha_1 h_n(s) - \Gamma_n(s).
\end{align*}
\end{proof}

\medskip
\noindent
\subsection{Proof of Equation \eqref{eq:chi-omega2}}\label{proof:chi} 
We will show that 
$$ \chi_{n}(s)=(h_n(s)- 2ns  ) \Gamma_n(s).$$
Consider  the expression of $w_i$ in \eqref{eq:w}. We find
\begin{align*}
\chi_n(s) & = \sum_{j=1}^n j w_j(s) = \left(\sum_{j=0}^{n-1} x_j(s)\right) \left(1+\sum_{\ell=0}^n \beta_{\ell} x_\ell(s) \right)  -  \left(\sum_{j=1}^{n} x_j(s)\right) \left(1+\sum_{\ell=0}^{n} \alpha_{\ell} x_{\ell-1}(s) \right).
\end{align*}
Using $e\cdot h = e\cdot v =0$ and $h_0=0$ yields
\begin{align*}
\sum_{j=0}^{n-1} x_j(s) &= \sum_{j=0}^{n-1} (1 + v_j s + h_j (s) )= n + ns - h_n(s), \\
\sum_{j=1}^{n} x_j(s) &= \sum_{j=1}^{n}  (1 + v_j s + h_j(s) ) = n -ns.
\end{align*}
Combined with Lemma~\ref{lem:gammaalpha}, this gives
\begin{align*}
\chi_n(s) & =\big(n + ns - h_n(s)\big)\big( \alpha_1 n - \alpha_1 n s   - \Gamma_n(s) \big) \\ & - \big(n -ns\big)\big( \alpha_1 n + \alpha_1 n s -  \alpha_1 h_n(s) - \Gamma_n(s)\big)  = (h_n(s)-2ns) \Gamma_n(s).
\end{align*}

\medskip
\noindent
\subsection{Proof of Lemma \ref{lem:useful}}\label{proof:useful}

\begin{proof}
Let $P_n(\ell)=n$ for $n$  even, and $P_n(\ell)=n-2\ell$ for $n$  odd. Then $\gamma_n(\ell)= \tfrac{(-1)^{\ell+1}  P_n(\ell)}{n(n+1)}\binom{n}{\ell}.$
This implies for $k\ge 0$,
\begin{align*}
\sum_{\ell=0}^{n}  \ell_{[k]} \gamma_{n}(\ell) &=
\sum_{\ell=k}^{n}  \ell_{[k]}  \frac{(-1)^{\ell+1}  P_n(\ell)}{n(n+1)}\binom{n}{\ell} = 
\tfrac{1}{n(n+1)}\sum_{\ell=k}^{n}  (-1)^{\ell+1}  P_n(\ell) \frac{\ell!}{(\ell-k)!}    \frac{n!}{\ell! (n-\ell)!} 
\\ & =\frac{(-1)^{k+1} n!}{n(n+1)(n-k)!}\sum_{j=0}^{n-k}  (-1)^{j}  P_n(j+k)  \binom{n-k}{j}.
\end{align*}
Since $P_n(j+k)$ has degree $0$ in $j$ for $n$ even and degree $1$ for $n$ odd, the identity \eqref{eq:binom0} implies 
$\sum_{\ell=0}^{n}  \ell_{[k]} \gamma_{n}(\ell)=0,$
for $n$ even and $n-k>0$, that is, $k<n$, and for $n$ odd and $n-k>1$, that is, $k<n-1$. This shows the first part of the statement. 

For $n$ even and $k=n$, the above reduces to
$$\sum_{\ell=0}^{n}  \ell_{[n]} \gamma_{n}(\ell)  =\frac{(-1)^{n+1} n!}{n(n+1)\cdot 0!} (-1)^{0}  P_n(n)  \binom{0}{0} = \tfrac{- n!}{n+1},$$
as in the statement.
For $n$ odd and $k=n-1$, it reduces to
$$\sum_{\ell=0}^{n}  \ell_{[n-1]} \gamma_{n}(\ell)  = \frac{(-1)^{n} n!}{n(n+1)\cdot 1!}\Big(
  P_n(n-1)  \binom{1}{0} - P_n(n)  \binom{1}{1}\Big)
=  \frac{-(n-1)!}{(n+1)}\left( (2-n)  - (-n)\right)=  \frac{-2(n-1)!}{(n+1)}.
 $$
This concludes the proof of the lemma.
\end{proof}
%

\begin{thebibliography}{10}

\bibitem{angeli06}
D.~Angeli and E.~D. Sontag.
\newblock Translation-invariant monotone systems, and a global convergence
  result for enzymatic futile cycles.
\newblock {\em Nonlinear Anal.-Real World Appl.}, 9(1):128--140, 2008.

\bibitem{banajiOscillation}
M.~Banaji.
\newblock Inheritance of oscillation in chemical reaction networks.
\newblock {\em Appl. Math. Comput.}, 325:191--209, 2018.

\bibitem{Carlson:1974uw}
D.~Carlson.
\newblock {A class of positive stable matrices}.
\newblock {\em J. Res. Natl. Bureau Stand.}, 78B:1--2, 1974.

\bibitem{carr81}
J.~Carr.
\newblock {\em {A}pplications of {C}entre {M}anifold {T}heory}.
\newblock Springer, Berlin, 1981.

\bibitem{cohen}
P.~Cohen.
\newblock The origins of protein phosphorylation.
\newblock {\em Nat. Cell Biol.}, 4(5):E127--E130, 2002.

\bibitem{calvin2017}
S.~Disselnkötter and A.~D. Rendall.
\newblock Stability of stationary solutions in models of the Calvin cycle.
\newblock {\em Nonlinear Anal.-Real World Appl.}, 34:481--494, 2017.

\bibitem{erdi-toth}
P.~\'{E}rdi and J.~T\'{o}th.
\newblock {\em {M}athematical Models of Chemical Reactions. Theory and
  Applications of Deterministic and Stochastic Models}.
\newblock Princeton University Press, Princeton, 1989.

\bibitem{walcher19}
E.~Feliu, S.~Walcher, and C.~Wiuf.
\newblock Noninteracting species and reduction of reaction networks.
\newblock {\em in preparation}.

\bibitem{Fel_elim}
E.~Feliu and C.~Wiuf.
\newblock Variable elimination in chemical reaction networks with mass-action
  kinetics.
\newblock {\em SIAM J. Appl. Math.}, 72:959--981, 2012.

\bibitem{feliu:intermediates}
E.~Feliu and C.~Wiuf.
\newblock Simplifying biochemical models with intermediate species.
\newblock {\em J. R. S. Interface}, 10:20130484, 2013.

\bibitem{feng:allosteric}
S.~Feng, M.~S\'aez, C.~Wiuf, E.~Feliu, and O.~S. Soyer.
\newblock {{C}ore signalling motif displaying multistability through
  multi-state enzymes}.
\newblock {\em J. R. S. Interface}, 13(123), 2016.

\bibitem{czechMath}
M.~Fiedler and V.~Pt\'{a}k.
\newblock On matrices with non-positive off-diagonal elements and positive
  principal minors.
\newblock {\em Czechoslovak Mathematical Journal}, 12:382--400, 1962.

\bibitem{flockerzi14}
D.~Flockerzi, K.~Holstein, and C.~Conradi.
\newblock {$n$}-site phosphorylation systems with {$2n-1$} steady states.
\newblock {\em Bull. Math. Biol.}, 76(8):1892--1916, 2014.

\bibitem{G-PNAS}
J.~Gunawardena.
\newblock {{M}ultisite protein phosphorylation makes a good threshold but can
  be a poor switch}.
\newblock {\em Proc. Natl. Acad. Sci. U.S.A.}, 102:14617--14622, 2005.

\bibitem{hell15a}
J.~Hell and A.~D. Rendall.
\newblock A proof of bistability for the dual futile cycle.
\newblock {\em Nonlinear Anal.-Real World Appl.}, 24:175--189, 2015.

\bibitem{hell16}
J.~Hell and A.~D. Rendall.
\newblock Sustained oscillations in the MAP kinase cascade.
\newblock {\em Math. Biosci.}, 282:162--173, 2016.

\bibitem{tau}
G.~V.~W. Johnson and W.~H. Stoothoff.
\newblock Tau phosphorylation in neuronal cell function and dysfunction.
\newblock {\em J. Cell Science}, 117(24):5721--5729, 2004.

\bibitem{feliu:unlimited}
V.~B. Kothamachu, E.~Feliu, L.~Cardelli, and O.~S. Soyer.
\newblock Unlimited multistability and boolean logic in microbial signalling.
\newblock {\em J. R. S. Interface}, (DOI: 10.1098/rsif.2015.0234), 2015.

\bibitem{kuehn15}
C.~Kuehn.
\newblock {\em Multiple time scale dynamics}, volume 191 of {\em Applied
  Mathematical Sciences}.
\newblock Springer, Cham, 2015.

\bibitem{kuznetsov95}
Y.~A. Kuznetsov.
\newblock {\em {E}lements of Applied Bifurcation Theory}.
\newblock Springer, Berlin, 1995.

\bibitem{Markevich-mapk}
N.~I. Markevich, J.~B. Hoek, and B.~N. Kholodenko.
\newblock {{S}ignaling switches and bistability arising from multisite
  phosphorylation in protein kinase cascades}.
\newblock {\em J. Cell Biol.}, 164:353--359, 2004.

\bibitem{perko}
L.~Perko.
\newblock {\em Differential equations and dynamical systems}, volume~7 of {\em
  Texts in Applied Mathematics}.
\newblock Springer-Verlag, New York, third edition, 2001.

\bibitem{ruiz}
S.~Ruiz.
\newblock An algebraic identity leading to {W}ilson's theorem.
\newblock {\em The Mathematical Gazette}, 80(489):579--582, 1996.

\bibitem{saez_reduction}
M.~S\'aez, C.~Wiuf, and E.~Feliu.
\newblock Graphical reduction of reaction networks by linear elimination of
  species.
\newblock {\em J. Math. Biol.}, 74:195--237, 2017.

\bibitem{p53}
E.~A. Slee, B.~Benassi, R.~Goldin, S.~Zhong, I.~Ratnayaka, G.~Blandino, and
  X.~Lu.
\newblock Phosphorylation of ser312 contributes to tumor suppression by p53 in
  vivo.
\newblock {\em PNAS}, 107(45):19479--19484, 2010.

\bibitem{TG-Nature}
M.~Thomson and J.~Gunawardena.
\newblock {{U}nlimited multistability in multisite phosphorylation systems}.
\newblock {\em Nature}, 460:274--277, 2009.

\bibitem{Wang:2008dc}
L.~Wang and E.~D. Sontag.
\newblock {On the number of steady states in a multiple futile cycle}.
\newblock {\em J. Math. Biol.}, 57(1):29--52, 2008.

\end{thebibliography}

\medskip
\paragraph{Acknowledgements. } EF and CW have been supported by the Independent Research Fund of Denmark. ADR is grateful for the hospitality of the 
Department of Mathematical Sciences, University of Copenhagen, while an important part of this work was being done.

\end{document}